%% file: main.tex
\documentclass[journal]{IEEEtran}
\IEEEoverridecommandlockouts
\usepackage{epsfig, endnotes}
\usepackage{amsmath}
\usepackage{amsthm}
\usepackage{amsmath,amsfonts}
\usepackage{algorithm}
\usepackage{algpseudocode}
\usepackage{array}
\usepackage{textcomp}
\usepackage{stfloats}
\usepackage{verbatim}
\usepackage{graphicx}
\usepackage{multirow}
\usepackage{balance}
\usepackage{graphicx} 
\usepackage{booktabs}
\usepackage{xcolor}
\usepackage{soul}
\usepackage{tabularx}
\usepackage{lipsum}  
\usepackage{amssymb}
\usepackage{bm}
\usepackage{mathtools}
\usepackage{amsthm}

\newtheorem{theorem}{Theorem}
\newtheorem{corollary}{Corollary}

\newtheorem{assumption}{Assumption}  
\theoremstyle{remark}

\usepackage{url}

\usepackage{wrapfig}
\usepackage{graphicx}
\usepackage{subfig}
\usepackage{circledtext}
\usepackage{pifont}
\AtBeginDocument{%
  }

\newcommand{\myparatight}[1]{\smallskip\noindent{\bf {#1}:}~}

\usepackage{algorithm}
\usepackage{algpseudocode}
\usepackage{multirow}
\usepackage{soul}

\usepackage{xspace}
\newcommand{\algname}{\texttt{OpenTwin}\xspace}

\newcommand{\cmark}{\ding{51}}
\newcommand{\xmark}{\ding{55}}
\newcommand{\pmark}{$\sim$}

\usepackage{cleveref} 
\crefname{axiom}{axiom}{axioms} 
\crefname{definition}{definition}{definitions}
\crefname{lemma}{lemma}{lemmata}

\allowdisplaybreaks[3]

\begin{document}

\setlength{\textfloatsep}{3pt}
\setlength{\intextsep}{3pt}
\setlength{\floatsep}{3pt}
\setlength{\jot}{3pt} 
\setlength{\abovedisplayskip}{3pt}
\setlength{\belowdisplayskip}{3pt}

\title{
\algname: 
Closed-Loop Digital Twins for Trustworthy Policy Deployment in Open RAN
}

\author{

		\IEEEauthorblockN{
            Zifan Zhang,
            Md Sharif Hossen,~\IEEEmembership{Senior Member,~IEEE},
            Dara Ron, \\
            Vijay K. Shah,~\IEEEmembership{Senior Member,~IEEE}, 
            Yuchen Liu,~\IEEEmembership{Member,~IEEE} \\
 }

\thanks{
Zifan Zhang and Yuchen Liu are with the Department of Computer Science, North Carolina State University, Raleigh, NC, 27695, USA (Email: \{zzhang66, yuchen.liu\}@ncsu.edu).}
\thanks{
Md Sharif Hossen, Dara Ron, and Vijay K. Shah are with the Department of Electrical and Computer Engineering, North Carolina State University, Raleigh, NC, 27695, USA (Email: \{mhossen, dron, vijay.shah\}@ncsu.edu).}
}

\maketitle

\begin{abstract}
In open radio access networks (O-RAN), the near-real-time RAN Intelligent Controller (RIC) hosts third-party xApps whose training and validation risk disrupting the operational network. Disaggregation amplifies this risk, as no single party can certify a control action end to end. Indeed, our testbed shows an E2 control request reported as successful while the base station never applies the change.
Digital twins (DTs) promise safe policy evaluation, yet existing O-RAN DTs largely rely on hand-crafted models, run without feedback from the deployment, and never say how often their predictions can be trusted.
To fill this gap, we present \algname, a closed-loop framework that learns the simulator configuration reproducing an operating deployment's streamed measurements, certifies the resulting DT by re-simulation, calibrates it online, and evaluates each xApp action before it executes on the physical network.
Every trust decision carries an error rate bounded by an operator-prescribed budget, with a drift detector limiting needless resynchronizations and a conformal fidelity gate admitting an action only when its predicted outcome range lies in the safe region.
Extensive experiments across simulation and real-world testbeds confirm single-digit percentage error in reproduced measurements, false approvals an order of magnitude below every budget from 0.05 to 0.30, and a gated energy-saving xApp that retains roughly 40\% of the saving achievable with perfect foresight.
\end{abstract}


\input{sections/intro}

\input{sections/background}
\input{sections/method}

\input{sections/theory}
\input{sections/eval}
\input{sections/exp}
\input{sections/conclusion}

\balance
\bibliographystyle{IEEEtran}
\bibliography{ref}

\end{document}

%% file: sections/intro.tex
\section{Introduction}

The next-generation open radio access network (O-RAN) is reshaping the cellular ecosystem and is being deployed in 5G/6G networks worldwide~\cite{tripathi2025fundamentals}.
Unlike vertically integrated legacy architectures, O-RAN promotes openness, programmability, and virtualization, allowing operators to mix multi-vendor hardware and software components.
Two RAN intelligent controllers (RICs) sit at the center of this design and enable software-defined, data-driven control~\cite{agarwal2025open}.
The near-real-time RIC (near-RT RIC) runs fast control loops between $10$~ms and $1$~s through xApps, while the non-real-time RIC (non-RT RIC) hosts rApps at slower timescales above $1$~s.
Both host third-party artificial intelligence (AI) microservices that address O-RAN use cases such as handover, radio resource optimization, traffic steering, and energy saving~\cite{agarwal2025open}.
Digital twins (DTs) have emerged as a complementary tool for developing these controllers.
A DT provides a continuously synchronized virtual copy of the physical network that supports monitoring, design, and analysis, and above all lets AI models be trained, verified, and validated without touching the live network~\cite{zhang2025digital}.
\mbox{ITU-T} has standardized network DT requirements and architecture~\cite{ITUT2022Y3090}, and ITU-R identifies the DT as a key 6G enabler toward 2030~\cite{ITU_IMT2030}.
However, a DT helps closed-loop control only if it mirrors the physical network closely and its outputs can be trusted, because a DT that drifts unnoticed can approve unsafe actions or trigger needless resynchronizations.

\myparatight{Motivation and Challenges}
Operators build O-RAN DTs to evaluate and deploy control policies without putting the operating network at risk.
We focus on xApps at the near-RT RIC, a harder setting than rApp policy tuning at the non-RT tier, because an xApp decision takes effect on the operating network within seconds and leaves no time for offline review.
In a vertically integrated RAN, a single vendor owns both the controllers and the base stations (BSs) and validates its own control actions internally.
O-RAN removes that ownership, because the near-RT RIC hosts third-party xApps whose developers typically have no access to the deployment they act on.
Their control commands also cross the E2 interface into a network whose components come from different vendors.
Thus, no single party owns both the controller and the controlled RAN, and none can certify a control action end to end.
An incorrect action therefore remains possible while no one is positioned to catch it.

Two preliminary measurements from our own testbed make this risk concrete.
In one open-source 5G stack, E2 control handlers can acknowledge success even when the target scheduler remains unchanged. Moreover, two E2 protocol versions interpret the same acknowledgment byte differently, creating a mismatch between reported and actual control execution (see Section~\ref{sub:interop} for details).
An xApp therefore cannot use the E2 interface to confirm what its own action did.
Trustworthy deployment requires a separate component that predicts the outcome of each action and verifies that prediction against measurements the deployment itself reports.
A DT can serve this role if it faithfully reproduces the deployment's key performance measurements (KPMs), evaluates an action within the time the near-RT loop allows, and reports how far its outputs can be trusted.
However, four obstacles make such a DT hard to build.
\textit{First}, KPMs remain scarce.
Operator traces are restricted by privacy and regulation, vendor measurement support is inconsistent, retention policies discard short-lived dynamics, and open-interface collection adds delay.
\textit{Second}, validating a policy on the operating network risks service disruption, as trial-and-error learning can cause instability and does not scale to large numbers of cells and users.
\textit{Third}, an xApp decision must be made within the near-RT window of at most one second, whereas one simulation run takes minutes, so candidate actions cannot be evaluated by launching a new simulation for each one.
\textit{Lastly}, observed KPMs do not point to a single configuration, since different configurations can produce similar measurements, and a DT built once and left unchanged drifts further from the deployment as conditions evolve.
A DT that reproduces the measurements offline and evaluates policies is therefore promising, but only if its accuracy can be certified with a known error rate.

\myparatight{Limitations and Opportunities}
Existing O-RAN DTs address parts of this open problem but leave a clear gap.
Ray-tracing platforms, channel emulators, and integrated software-stack DTs~\cite{NVIDIA_AerialOmniverse2024, polese2024colosseum, owdt2025} reproduce a deployment from a hand-built site model or a pre-generated scenario, which serves pre-deployment testing well.
None of them calibrates against the operating deployment while it runs, and none reports how often its verdict on an individual action is wrong.
Platforms that do connect to an operating deployment either copy its live measurements into a cloud model for what-if analysis or standardize KPM collection for xApp training~\cite{deshpande2025twinran, 10.1145/3592149.3592161}.
In both cases the decision to act still rests on confidence that was never measured.
A faithful DT alone does not close the gap either, since a DT can match every observed KPM and still give no evidence about how often its verdict on a specific action is wrong.
Distribution-free inference~\cite{simeone2025conformal, hou2024automatic} provides guarantees that hold for the amount of data actually collected and require no assumption about how KPMs are distributed, yet no existing closed-loop O-RAN DT attaches such guarantees to its trust decisions.
We call this the \textit{safe-control gap}, since no party or interface can certify a control action before it executes.
Closing it turns the DT from a passive model into a trusted component of the control loop with a prescribed error budget.

\myparatight{Contributions}
To this end, we develop \textbf{\algname}, a high-fidelity closed-loop DT for validating and verifying xApps at the near-RT RIC.
\algname restores the end-to-end certification that disaggregation removed and keeps the error rate of every trust decision below a budget the operator sets in advance.
In essence, three core techniques are developed:

\begin{itemize}
\item \textit{First}, \algname constructs the DT directly from telemetry, without a hand-built site model.
A learned inverse model maps the live KPM stream to the simulator configuration that would reproduce it, and re-simulation over candidate configurations certifies the DT against the same measured stream.
Online calibration and drift-triggered resynchronization then keep the DT aligned with the deployment as conditions evolve.
Note that \algname is not tied to a particular simulator, since the DT can be instantiated on any network simulator that exposes the same configuration parameters.

\item \textit{Second}, a sequential drift detector monitors the deviation between calibrated DT outputs and physical measurements, and it triggers a resynchronization only once the accumulated evidence that the DT no longer matches the deployment crosses a threshold.
Its bounded detection delay and the calibrator's bounded tracking error together bound the fidelity duty cycle, the fraction of time the DT runs less accurately.

\item \textit{Third}, a conformal fidelity gate turns the DT's recent prediction errors into a range of plausible post-action KPMs, a prediction set, that contains the true outcome at a rate the operator prescribes. A calibrated forward model predicts these KPMs without additional simulation, and the gate admits an xApp action only when its prediction set lies within the operator-defined safe region.
A theoretical analysis proves guarantees that hold at the sample sizes actually available, bounding how often the detector triggers an unnecessary DT resynchronization under any inspection schedule and keeping unsafe approvals below the prescribed budget even as the environment drifts (see Section~\ref{sec:theory}).

\item Extensive experiments spanning simulation, a multi-vendor 5G testbed, and real radio hardware validate \algname.
They reveal interoperability defects in E2 control and demonstrate DT-gated xApp control, with actuation verified at the RAN scheduler.

\end{itemize}

To the best of our knowledge, \algname is the first closed-loop O-RAN DT that governs both resynchronization and action admission under error budgets prescribed in advance.



%% file: sections/background.tex
\section{Background and Related Work}
\label{sec:background}
\label{sec:related}

\begin{table*}[t]
\centering
\caption{Positioning of \algname vs. O-RAN DT platforms (\cmark{} native, \pmark{} partial, \xmark{} absent).}
\label{tab:positioning}
\renewcommand{\arraystretch}{1.0}\footnotesize\setlength{\tabcolsep}{3.5pt}
\begin{tabular}{lcccccccc}
\toprule
\textbf{Platform} & \textbf{System reproduced} & \textbf{Config.} & \textbf{Scenario} & \textbf{Online KPM} & \textbf{Drift-triggered} & \textbf{Control} & \textbf{Decision} & \textbf{O-RAN} \\
 & \textbf{by the DT} & \textbf{inference} & \textbf{portability} & \textbf{calibration} & \textbf{resync} & \textbf{validation} & \textbf{guarantee} & \textbf{interfaces} \\
\midrule
Aerial Omniverse~\cite{NVIDIA_AerialOmniverse2024} & Ray-traced channel & \xmark & \xmark & \xmark & \xmark & \pmark & \xmark & \pmark \\
Colosseum DT~\cite{polese2024colosseum} & RF channel emulator & \xmark & \xmark & \pmark & \xmark & \pmark & \xmark & \cmark \\
TwinRAN~\cite{deshpande2025twinran} & Operating RAN via E2 & \xmark & \pmark & \pmark & \xmark & \pmark & \xmark & \cmark \\
Open Wireless DT~\cite{owdt2025} & OAI + ray tracing & \xmark & \xmark & \xmark & \xmark & \xmark & \xmark & \cmark \\
ns-O-RAN-flexRIC~\cite{ns-O-RAN-flexric} & Packet-level simulation & \xmark & \xmark & \xmark & \xmark & \pmark & \xmark & \cmark \\
\addlinespace[7pt]
\textbf{\algname (this work)} & Operating RAN, software to RF & \cmark & \pmark & \cmark & \cmark & \cmark & \cmark & \cmark \\
\bottomrule
\end{tabular}
\vspace{-.2in}
\end{table*}

\myparatight{DT platforms for O-RAN}
To date, several platforms build DTs for wireless and O-RAN systems~\cite{zhou2025digital}, and Table~\ref{tab:positioning} compares \algname with them across seven criteria.
Every platform in the table reproduces its target system from a hand-built model, whether ray-traced, emulated, or simulated over open-source RAN and RIC software~\cite{NVIDIA_AerialOmniverse2024, polese2024colosseum, owdt2025, ns-O-RAN-flexric}, and no surveyed platform learns its configuration parameters from telemetry.
TwinRAN~\cite{deshpande2025twinran} copies live E2 measurements into a cloud model, which re-targets the DT to a new deployment but never recovers the configuration parameters that produced those measurements, hence its partial marks in the table.
None continuously calibrates its DT against the live KPM stream or regenerates it once drift shows that it no longer explains the deployment.
Recent work~\cite{zhang2026untwinning} targets the same maintenance need from the reverse direction, searching for the configuration that explains a twin after it has drifted.
In the control-validation column, a partial mark denotes pre-deployment testing rather than per-action screening, and no platform certifies a policy against a guaranteed error rate.
In contrast, \algname builds its DT from the operating RAN itself through standardized KPMs, and it closes the safe-control gap by learning the configuration of that specific deployment and validating control policies in the DT before they execute on the physical network.
Scenario portability, the ability to re-target a DT to a new site or deployment, remains partial for \algname, since a new deployment requires re-running the inversion on its measurements and we have not yet evaluated that transfer across sites.

\myparatight{Closed-loop control and safety in O-RAN}
The ns-O-RAN extension of OpenRAN Gym~\cite{10.1145/3592149.3592161} connects ns-3 with an E2-compliant RIC for machine learning data collection and hosts traffic-steering xApps~\cite{lacava2023programmable}.
SafetyCopilot~\cite{jawad2026safetycopilot} screens near-RT RIC actions against safety rules written by hand, whereas \algname admits an action only when a calibrated range of predicted outcomes stays inside the safe region, at a known error rate.

\myparatight{Guarantees for learning-driven control}
The two trust decisions of \algname build on two lines of distribution-free inference, online conformal prediction and sequential testing.
Adaptive conformal inference keeps a prediction set covering a target fraction of outcomes even when the data distribution shifts~\cite{gibbs2021adaptive, gibbs2022conformal}, and is the basis of our conformal fidelity gate.
Conformal risk control extends this coverage guarantee to losses that change monotonically with the decision threshold~\cite{angelopoulos2022conformal}, the property our false-approval bound relies on.
E-detectors accumulate evidence over time and bound a monitor's false-alarm rate no matter how often it is inspected~\cite{shin2022edetectors}, and they underlie our resynchronization detector.
In wireless systems, this line of work has calibrated the outputs of black-box models~\cite{simeone2025conformal} and used a DT to select among candidate models online~\cite{hou2024automatic}, and the latter is the closest setting to ours.
Tao~\emph{et~al.}~\cite{tao2025provable} bound how far a policy trained in a DT can fall below optimal once deployed.
Their bound uses a bisimulation metric, a state-similarity measure estimated offline from sampled transitions, so it does not use the DT-to-physical deviation an operating loop observes at run time.
Across these works, the guarantee attaches to a predictor's output, a model choice, or an offline policy comparison, whereas our \algname applies it to two decisions inside a DT-driven loop, namely whether to admit an action and whether to resynchronize the DT.

%% file: sections/method.tex
\section{Overview of \algname System}
\label{sec:framework}
\label{sec:sysarch}

The DT denotes the maintained virtual representation of one operating deployment rather than any particular simulator.
It comprises a network simulation engine, hereafter the backbone simulator, that regenerates deployment measurements from learned configuration parameters.
An online calibrator then keeps the simulator's outputs aligned with that deployment as each new KPM report arrives.
A calibrated forward model completes the DT and scores candidate xApp actions before any reaches the physical network.
\algname constructs and maintains this DT at the non-RT tier, yet it guards xApp actions that execute at the near-RT RIC.
Hosting the DT at the slower tier adds no latency to that path, since an admission check costs only microseconds.

Figure~\ref{fig:framework} shows the overall architecture.
A single RIC instance connects the DT to the physical network, the physical network reports standardized KPMs through an E2 agent, and the loop consists of three modules.
\textbf{The first module}, DT construction (Section~\ref{sec:construction}, steps {\Large \ding{172}}--{\Large \ding{175}}), consumes the raw KPM stream and produces a running, certified DT.
An inverse model learns configuration parameters that reproduce the observed KPMs, re-simulation then checks the instantiated DT against the observed telemetry, and the certified DT regenerates its own KPM stream.
\textbf{The second module}, online calibration for fidelity and freshness (Section~\ref{sec:calibration}, steps {\Large \ding{176}}, {\Large \ding{178}}, and {\Large \ding{179}}), consumes the physical and DT-side KPM streams and produces calibrated DT outputs together with a deviation score that measures how far the DT has drifted from the physical network.
A resynchronization detector watches that score and re-inverts the configuration parameters when the DT no longer explains the deployment.
\textbf{The third module}, trustworthy action admission (Section~\ref{sec:trust}, steps {\Large \ding{177}} and {\Large \ding{180}}), consumes candidate xApp actions together with the forward model's post-action predictions and produces the admission decision.
A supervisor engine hosts both decision stages, marked by the dashed grouping in Fig.~\ref{fig:framework}.
An admitted action returns along the E2 control path, closing the loop, and Algorithm~\ref{alg:OpenTwin} summarizes one iteration.
The budgets $\alpha_{\mathrm{res}}$ and $\alpha_{\mathrm{gate}}$ are the only error budgets an operator sets in advance, and each is an allowed error rate rather than a threshold on the deviation score.
Note that \algname does not depend on which 5G implementation registers over E2, and Section~\ref{sub:interop} demonstrates this by attaching a second independently implemented stack through a translation layer that rewrites its E2 messages into the encoding the shared RIC expects.

\begin{figure*}[t]
\centering
\includegraphics[width=0.9\textwidth]{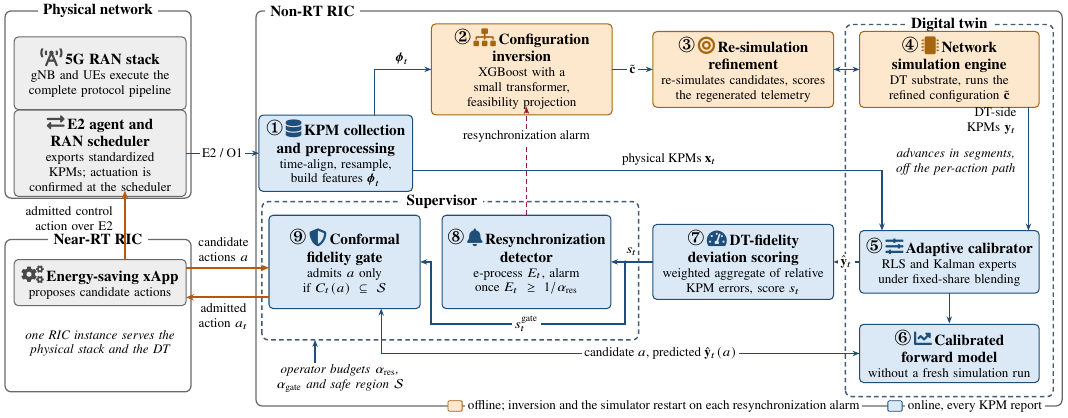}
\caption{The \algname closed-loop architecture. Steps {\Large \ding{172}}--{\Large \ding{180}} mark the loop stages detailed in Sections~\ref{sec:construction}--\ref{sec:trust}, and the dashed box marks the supervisor engine hosting the two trust decisions.}
\label{fig:framework}
\vspace{-.1in}
\end{figure*}


\begin{algorithm}[t]
\caption{\algname Closed-Loop Framework}
\label{alg:OpenTwin}
\begin{algorithmic}[1]
\State {\bfseries Input:} KPM stream $\{\mathbf{x}_t\}$, budgets $\alpha_{\mathrm{res}},\alpha_{\mathrm{gate}}$, safe region $\mathcal{S}$, calibrator parameters $(\lambda,W,\alpha_{\text{sh}},\eta)$, step size $\gamma$, declared score range $B$, shift grid $(\{\delta_k\},\{\pi_k\})$
\State {\bfseries Init:} $E_0\!\leftarrow\!1$, radius $q_1$ from the calibration window, calibrator state reset, configuration $\tilde{\mathbf{c}}$ from the initial inversion, certified by re-simulation (Section~\ref{sec:refine})
\For{$t=1$ {\bfseries to} $T$}
    \If{$E_{t-1} \ge 1/\alpha_{\mathrm{res}}$}
        \State Re-invert and project fresh configuration parameters $\tilde{\mathbf{c}}$ from recent KPMs (Section~\ref{sec:config-infer}), restart the backbone simulator, and reset the calibrator and detector
    \EndIf
    \State $\mathbf{y}_t \leftarrow \text{sim}(\tilde{\mathbf{c}})$; form the calibrated estimate $\hat{\mathbf{y}}_t$ per KPM with the fixed-share expert set
    \State Form the deviation score $s_t$ by~\eqref{eq:global_dev} and update the e-process $E_t$ by~\eqref{eq:eprocess}
    \For{each candidate xApp action $a$}
        \State Predict $\hat{\mathbf{y}}_t(a)$ with the forward model and admit $a$ if the prediction set $C_t(a)$ of~\eqref{eq:confset} at radius $q_t$ lies inside $\mathcal{S}$
    \EndFor
    \State Apply the xApp's preferred admissible action $a_t$, or the null action if none, observe $\mathbf{x}^{\mathrm{post}}_t$, score it as $s^{\mathrm{gate}}_t=\varrho(\mathbf{x}^{\mathrm{post}}_t,\hat{\mathbf{y}}_t(a_t))$, and update the radius $q_{t+1}$ by~\eqref{eq:aci}
\EndFor
\end{algorithmic}
\end{algorithm}

\section{O-RAN DT Construction}
\label{sec:construction}

\subsection{KPM Preprocessing}
\label{sec:kpm-preproc}
Raw telemetry arrives over E2 at non-uniform intervals and heterogeneous scales.
KPM collection and preprocessing (step {\Large \ding{172}} in Fig.~\ref{fig:framework}) therefore produces two time-aligned outputs, the feature vector $\boldsymbol{\phi}_t$ that the inverse model consumes and the physical KPM vector $\mathbf{x}_t$ that the calibrator regresses the DT output onto.
Let $\mathcal{M}$ denote the set of monitored KPMs and $\mathbf{x}_t=[x^{(m)}_t]_{m\in\mathcal{M}}$ the physical KPM vector at index $t$, where timestamps align to a common clock and every stream is resampled to sampling interval $T_s$.
We drop the superscript $m$ when discussing a generic scalar KPM.
Each KPM then contributes lagged terms $\{x_{t-1},\dots,x_{t-L}\}$ of order $L$, a windowed average $\bar{x}_t = \frac{1}{W}\sum_{i=0}^{W-1} x_{t-i}$, rate-of-change terms, and cross-KPM ratios.
Every feature is standardized with training-set statistics and carries a per-sample quality mask $b_t\in\{0,1\}$ that flags unreliable values, while missing values and heterogeneous vendor reporting periods are handled by the synchronized DT, which regenerates missing measurements.

\subsection{Network Configuration Inversion}
\label{sec:config-infer}
Given preprocessed features, the XGBoost-based~\cite{chen2016xgboost} model of Fig.~\ref{fig:framework} (step {\Large \ding{173}}) infers configuration parameters that reproduce the observed KPMs, and those parameters initialize the DT.
We denote the set of configuration parameters as a mixed vector $\mathbf{c}_t=(\mathbf{c}^{\text{cont}}_t,\mathbf{c}^{\text{cat}}_t)$, where $\mathbf{c}^{\text{cont}}_t\in\mathbb{R}^{p}$ collects bounded continuous parameters such as bandwidth, packet size, and buffer size, each with box constraints $[\ell_j,u_j]$, and $\mathbf{c}^{\text{cat}}_t$ collects categorical parameters such as mobility model or handover policy, each drawn from a vocabulary $\mathcal{V}_k$.
Boosted trees re-fit quickly on the simulated training set and evaluate cheaply at every resynchronization.
Because the backbone simulator generates labeled examples on demand, we train the XGBoost models of the inverse map on simulated pairs $\{(\boldsymbol{\phi}^{\text{sim}}_s,\mathbf{c}_s)\}_{s\in\mathcal{T}}$, where $\boldsymbol{\phi}^{\text{sim}}_s$ is the feature vector of a simulated KPM sample and $\mathbf{c}_s$ are its known configuration parameters.
Specifically, the regressors $\{f_j\}_{j=1}^{p}$ and classifiers $\{g_k\}_{k=1}^{K_{\mathrm{cat}}}$ are trained under one composite loss:
\predisplaypenalty=100
\begin{equation}
\begin{aligned}
\mathcal{L} = \sum_{s\in\mathcal{T}} w_s \Bigg[
    & \sum_{j=1}^{p} a_j \mathcal{L}_{\text{Huber}}
    \!\big(f_j(\boldsymbol{\phi}^{\text{sim}}_s),\, c^{\text{cont}}_{j,s}\big) \\
    & + \sum_{k=1}^{K_{\mathrm{cat}}} a^{\text{cat}}_k \mathcal{L}_{\text{CE}}
    \!\big(g_k(\boldsymbol{\phi}^{\text{sim}}_s),\, c^{\text{cat}}_{k,s}\big)
    \Bigg],
\end{aligned}
\label{eq:composite}
\end{equation}
where $a_j,a^{\text{cat}}_k>0$ weight each parameter by how much an error on it costs the DT, the Huber loss guards against outliers, and the cross-entropy loss uses class weights to mitigate imbalance.
A small Transformer encoder serves as a second estimator for report periodicity and handover mode, since these two parameters show up in the time ordering of a KPM stream rather than in its window aggregates.
Table~\ref{tab:identifiability} confirms that each estimator recovers a different kind of parameter, since the Transformer leads only on these two and the boosted trees match or beat it on every other identifiable one.
Neither changes what the KPMs expose, as the error of both runs from a few percent on the most identifiable parameters to above 100\% on the buffer size.
\algname therefore adds the calibration stage of Section~\ref{sec:calibration} rather than trusting inversion alone.
Gradient-boosted trees consume window aggregates, so they miss configuration effects that appear only in the temporal shape of a KPM sequence.
Report periodicity, in contrast, follows almost linearly from the number of KPM indications a run emits, while handover-timer mode alters downlink scheduling even without a handover.
The sample weight combines data quality and temporal recency as $w_s = b_s\,\mu^{\,T-s}$ with $\mu\in(0,1)$, where $b_s$ is the quality mask, $T$ the most recent training index, and $\mu$ a recency decay.
Both stay uniform in the reported fits.

At inference, the physical feature vector $\boldsymbol{\phi}_t$ is passed through the trained models to obtain continuous outputs $\hat{\mathbf{c}}^{\text{cont}}_t=[f_1(\boldsymbol{\phi}_t),\dots,f_p(\boldsymbol{\phi}_t)]^{T}$ and, for each categorical parameter, a probability vector $\hat{\boldsymbol{\pi}}^{(k)}_t=g_k(\boldsymbol{\phi}_t)$.
The raw outputs are then projected onto a feasible configuration to guarantee schema validity, where the continuous parameters solve a small quadratic program,
\begin{equation}
\tilde{\mathbf{c}}^{\text{cont}} = \arg\min_{\mathbf{z}} \;\|\mathbf{z}-\hat{\mathbf{c}}^{\text{cont}}\|_2^2 \quad \text{s.t.}\; \ell_j \le z_j \le u_j,\;\; \mathbf{A}\mathbf{z}\le \mathbf{b}, \nonumber
\end{equation}
and categorical parameters take the most probable label, i.e., $\tilde{c}^{\text{cat}}_k=\arg\max_{v\in\mathcal{V}_k}\hat{\pi}^{(k)}_{t,v}$.

\begin{table}[t]
\centering
\caption{Configuration inversion by identifiability class under the group-by-run split over the 300-run sweep, in mean absolute percentage error (MAPE, \%) for continuous parameters and exact-match rate for categorical ones, as five-seed means.
Bold marks the better estimator, or both where their MAPE agrees to within one point.}
\label{tab:identifiability}
\renewcommand{\arraystretch}{1.0}\footnotesize
\begin{tabular}{llcc}
\toprule
\textbf{Class} & \textbf{Parameter} & \textbf{XGBoost} & \textbf{Transformer} \\
\midrule
\multirow{3}{*}{Strong} & Packet size & \textbf{4.4} & 26.7 \\
 & Bandwidth & \textbf{10.8} & 24.3 \\
 & Center frequency & \textbf{13.6} & \textbf{14.0} \\
\midrule
\multirow{6}{*}{Moderate} & Inter-site dist.\ (UE) & \textbf{18.7} & 27.8 \\
 & Report periodicity & 23.8 & \textbf{11.4} \\
 & Outage threshold & \textbf{32.1} & \textbf{32.5} \\
 & Inter-site dist.\ (cell) & \textbf{35.5} & \textbf{35.3} \\
 & Handover SINR margin & 39.0 & -- \\
 & Max UE speed & \textbf{40.4} & \textbf{39.6} \\
\midrule
Weak & Min UE speed & \textbf{86.2} & 89.3 \\
\midrule
Non-ident. & Buffer size & 146.6 & \textbf{126.3} \\
\midrule
\multirow{2}{*}{Categorical} & Mobility model & \textbf{0.59} & 0.36 \\
 & Handover mode & 0.86 & \textbf{0.99} \\
\bottomrule
\end{tabular}
\end{table}

\myparatight{Configuration parameter identifiability}
The inferred configuration $\tilde{\mathbf{c}}$ is the projected output of the trained models and need not be unique.
The KPM-to-configuration map is only partially identifiable, meaning that some parameters cannot be recovered from logged KPMs at all.
Table~\ref{tab:identifiability} partitions the parameters into four identifiability classes, strong, moderate, weak, and non-identifiable, by how strongly a logged KPM responds to each.
Two mechanisms push a parameter toward the weaker classes:
1) A parameter whose control rule never activates has no effect on the KPMs and therefore cannot be inferred from them.
The handover SINR margin behaves this way under time-to-trigger modes that never reach the threshold comparison it governs, so its entry in Table~\ref{tab:identifiability} comes from a separate sweep in which that rule is active.
2) A parameter whose effect never appears under the offered load is equally invisible.
The configured buffer size behaves this way under saturated traffic, since observed occupancy stays far below the configured maximum, while weak parameters do affect the KPMs and the offered load only partly masks that effect.
Because the map is only \textit{partially} identifiable, \algname does not treat the inferred configuration as correct by itself and instead relies on the guarantees of Section~\ref{sec:theory}.

\subsection{Telemetry Certification by Re-Simulation}
\label{sec:refine}
The inverse map is trained on simulated pairs, so its initial output must be certified against the deployment's own telemetry (step {\Large \ding{174}}) before the DT goes live.
This stage therefore re-simulates candidate configurations, scores each candidate's KPM vector against the KPMs observed from the RAN with the deviation score of Section~\ref{sec:dev-score} formed from raw deviations, and keeps the candidate with the lowest score.
The search is motivated by covariance-matrix adaptation~\cite{hansen2016cmaes}.
It scales per-coordinate step sizes by identifiability class and evaluates every candidate under the same fixed random seed across generations, so an improvement carried across generations is not an artifact of the random draw.\footnote{In practice, the search plateaus within one or two generations, about twenty re-simulation runs.}
This stage certifies that the DT reproduces the observed telemetry rather than recovering the true configuration, because many configurations reproduce the same KPMs, as Section~\ref{sec:results} quantifies.

\section{Online DT Calibration for Fidelity and Freshness}
\label{sec:calibration}
\subsection{Adaptive DT-Fidelity Calibrator}
\label{sec:rls}
The backbone simulator produces KPMs under inferred configuration parameters (step {\Large \ding{175}} in Fig.~\ref{fig:framework}), and those KPMs differ from physical measurements even at start-up, since the simulator represents the radio environment in simplified form.
This gap between simulation and reality widens as the physical network evolves, so \algname continuously calibrates DT output toward the physical measurements (step {\Large \ding{176}}).
To synchronize the DT at this KPM level, we run a small set of lightweight calibration filters in parallel, called experts in online learning, and blend their outputs with weights that shift toward whichever expert has recently matched the physical KPMs best.
Specifically, recursive least squares (RLS) with forgetting needs few samples, which suits the low rate at which the E2 interface reports KPMs, and it tracks the gradual load and mobility evolution that dominates steady operation.
The Kalman expert instead re-adapts within a few samples when a traffic change or handover causes an abrupt shift.
The set is extensible, yet these two experts cover the regimes O-RAN telemetry exhibits.

\myparatight{RLS expert}
The first expert calibrates the digital-to-physical residual with an RLS filter.
Let $x_t$ denote a physical KPM and $y_t$ the corresponding raw DT output.
We approximate the time-varying relation $x_t\approx g(y_t,y_{t-1},\dots,y_{t-L},\bar{y}_t)$ with a linear adaptive filter using the feature vector
\begin{equation}
\mathbf{u}_t = \big[\,1,\, y_t,\, y_{t-1},\dots, y_{t-L},\, \bar{y}_t \big]^{T} \in \mathbb{R}^{L+3},
\end{equation}
and one-step calibrated estimate $\mathbf{u}_t^{T}\boldsymbol{\theta}_{t-1}$, where $\boldsymbol{\theta}_t\in\mathbb{R}^{L+3}$ is the time-varying coefficient vector and $\bar{y}_t$ is the $W$-sample moving average of DT outputs, formed as in Section~\ref{sec:kpm-preproc}.
With forgetting factor $\lambda\in(0,1)$ and inverse-covariance matrix $\mathbf{P}_t$, the RLS updates are
\begin{align}
\mathbf{k}_t &= \frac{\mathbf{P}_{t-1}\mathbf{u}_t}{\lambda + \mathbf{u}_t^{T}\mathbf{P}_{t-1}\mathbf{u}_t}, \\
\boldsymbol{\theta}_t &= \boldsymbol{\theta}_{t-1} + \mathbf{k}_t \big(x_t - \mathbf{u}_t^{T}\boldsymbol{\theta}_{t-1}\big), \\
\mathbf{P}_t &= \frac{1}{\lambda}\!\left(\mathbf{P}_{t-1} - \mathbf{k}_t \mathbf{u}_t^{T} \mathbf{P}_{t-1}\right),
\end{align}
initialized with a regularization term via $\mathbf{P}_0=\rho^{-1}\mathbf{I}$ to stabilize early updates when inputs are correlated.
Equivalently, $\boldsymbol{\theta}_t$ minimizes the exponentially weighted, $\ell_2$-regularized objective $\sum_{s=1}^{t}\lambda^{t-s}(x_s-\mathbf{u}_s^{T}\boldsymbol{\theta})^2+\lambda^{t}\rho\|\boldsymbol{\theta}\|_2^2$, applied per KPM in parallel.
The forgetting factor $\lambda$ governs the speed-stability trade-off that Section~\ref{sec:theory} analyzes.

\myparatight{Kalman expert}
The second expert uses a Kalman filter to track the scale and offset between DT output and physical measurement as a random walk.
It re-adapts to a level shift faster than exponential forgetting, because a step change enters directly through its process-noise covariance.
It models the digital-to-physical relation as $x_t=\alpha_t y_t+\beta_t+\varepsilon_t$, where the scale-offset pair $(\alpha_t,\beta_t)$ evolves as a random walk and $\varepsilon_t$ is the observation noise.
Its recursion mirrors the gain, coefficient-update, and covariance steps written above for RLS, with the forgetting factor $\lambda$ replaced by an explicit process-noise covariance that sets how fast $(\alpha_t,\beta_t)$ may move and by the variance of $\varepsilon_t$ in the gain denominator.
RLS with exponential forgetting can itself be viewed as a Kalman filter whose process noise is induced by $\lambda$, so the two experts share the same tracking-error structure and differ only in how that process noise is set.

\myparatight{Fixed-share aggregation}
In plain terms, the calibrator tracks which expert has recently predicted the physical KPMs best and shifts its blend toward that expert, without ever discarding any expert entirely.
At each step, every expert emits its one-step-ahead estimate before the physical KPM $x_t$ is revealed, and the calibrator emits the weight-blended estimate $\hat{y}^{\text{b}}_t=\sum_k w_{k,t}\,\hat{y}^{(k)}_t$, written $\hat{y}_t$ hereafter.
When $x_t$ arrives, each expert is scored by the loss of the estimate it had already issued, and the weights are updated with a fixed-share rule~\cite{orabona2026online} adapted to our per-KPM setting.
A Hedge update $w_k\!\leftarrow\!w_k\exp(-\eta h_k)$ acts on a loss $h_k$ rescaled to $[0,1]$ within that step, and a share step then redistributes a fraction $\alpha_{\text{sh}}$ of each weight uniformly to the others, so $\alpha_{\text{sh}}$ sets the switching rate.
The share step keeps every weight at or above $\alpha_{\text{sh}}/(N-1)>0$, where $N$ is the number of experts, so an expert that had fallen to negligible weight can regain the lead within a number of favorable steps that grows only logarithmically in $1/\alpha_{\text{sh}}$.
Each loss is divided by a running scale estimate shared across that step's experts and truncated at one before the exponential update, so a single learning rate $\eta$ works for every KPM.

\subsection{DT-Fidelity Deviation Score}
\label{sec:dev-score}
The DT-fidelity deviation scoring stage of Fig.~\ref{fig:framework} (step {\Large \ding{178}}) forms the deviation score that both trust decisions use.
Resynchronization applies the score to the calibrated residual, and admission applies it to the residual measured after an action has run.
In conformal prediction terms this is a nonconformity score, a number that measures how unusual a new observation looks against past ones, and a larger deviation means a larger mismatch between the DT and the physical network.
It is formed from relative errors of the calibrated outputs, scaled by $100$ so that scores read in percent,
\begin{equation}
e^{\text{cal}}_{m,t}
= 100\,\frac{\hat{y}^{(m)}_t - x^{(m)}_t}{\max(|x^{(m)}_t|,\epsilon_0)},
\end{equation}
with $\epsilon_0>0$ guarding against small denominators.
The raw counterpart $e^{\text{raw}}_{m,t}$ replaces $\hat{y}^{(m)}_t$ with the uncalibrated output $y^{(m)}_t$.
The score combines the absolute per-KPM deviations $d_{m,t}$ under nonnegative weights, with $d_{m,t}=e^{\text{cal}}_{m,t}$ for the online loop and $d_{m,t}=e^{\text{raw}}_{m,t}$ when certifying construction,
\begin{equation}
s_t = \frac{\sum_{m} w_m\,|d_{m,t}|}{\sum_{m} w_m},
\label{eq:global_dev}
\end{equation}
where a practical choice sets $w_m$ proportional to an operator-chosen priority divided by the running scale $\widehat{\sigma}_m$, an estimate of the typical magnitude, so that inherently high-variance KPMs do not dominate.
The monitored set includes the cell-load and SINR KPMs, so the same score measures both DT fidelity and the behavior the safe region constrains.

\subsection{Resynchronization Detector}
\label{sec:edetector}
A persistently large deviation score signals \textit{drift}, i.e., a configuration or environment change that moves the deployment away from the behavior the DT reproduces.
After such a change, the inferred configuration parameters no longer explain the physical network.
The DT should then be resynchronized, which re-inverts the configuration parameters, restarts the backbone simulator, and resets the calibrator and the detector.
The supervisor engine inspects the score after every KPM indication, and a monitor with a fixed threshold would raise ever more false alarms under such frequent checking.
We therefore design a resynchronization detector (step {\Large \ding{179}}) that accumulates evidence against the DT in a running statistic $E_t$, an e-process~\cite{shin2022edetectors}.
This statistic stays small while nothing changes, grows once the deviation stream drifts, and can be inspected at any moment without inflating the false-alarm rate, a property termed anytime-validity in sequential testing~\cite{ramdas2023gametheoretic}.
We first standardize the deviation stream to $z_t=(s_t-\mu_0)/\sigma_0$ using robust statistics gathered over a period that contains no genuine change, a regime we call in-control.
We then construct a nonnegative mixture e-process over a grid of hypothesized one-sided shifts $\{\delta_k\}_{k=1}^{K}$ with weights $\{\pi_k\}$ summing to one, one-sided since only an increase in deviation signals drift,
\begin{equation}
E_t = \sum_{k=1}^{K} \pi_k \prod_{r=\nu_0+1}^{t}\!\exp\!\big(\delta_k z_r - \delta_k^2/2\big),
\quad E_{\nu_0}=1,
\label{eq:eprocess}
\end{equation}
where each factor is the likelihood ratio of a mean shift of size $\delta_k$ against the no-change hypothesis $H_0$, and $\nu_0$ is the last reset.
Under $H_0$ the process $\{E_t\}$ never grows in expectation, $\mathbb{E}_{H_0}[E_t]\le 1$, and the detector raises a resynchronization alarm the first time it crosses the threshold, at $\tau = \inf\{t>\nu_0 : E_t \ge 1/\alpha_{\mathrm{res}}\}$.
Here an epoch is the stretch between two detector resets.
Since $E_t$ only \textit{rarely} reaches a large value, at most an $\alpha_{\mathrm{res}}$ fraction of in-control epochs ends in a false resynchronization under any inspection schedule, as Section~\ref{sec:theory} formalizes.
The mixture over the shift grid $\{\delta_k\}$ removes the need to know the post-change magnitude.

\section{Trustworthy Action Admission}
\label{sec:trust}
\label{sec:gate}
Once the DT is constructed and continuously synchronized, an xApp action can be admitted only after it is checked to be safe, so the gate (step {\Large \ding{180}}) guards the E2 control path in Fig.~\ref{fig:framework}.
We therefore calibrate an online conformal quantile of the deviation score, a threshold learned from recent scores so that all but a prescribed fraction fall below it.
Throughout, coverage denotes how often the calibrated prediction set contains the realized outcome, the fraction that the target $1-\alpha_{\mathrm{gate}}$ prescribes.
This is a statistical quantity unrelated to radio coverage.
The DT advances in segments, where a segment denotes one block of simulated time that the backbone simulator executes without pausing and that carries at most one applied control action.
The safe region $\mathcal{S}$ is the set of post-action KPM vectors declared acceptable for deployment, such as every served UE staying above an SINR floor and cell load staying below a cap.
Inspired by adaptive conformal inference (ACI)~\cite{gibbs2021adaptive}, we track the calibrated quantile online on the gate's own post-action deviation $s^{\mathrm{gate}}_t$ and use it as the radius of the prediction set as below:
\begin{equation}
q_{t+1} = \max\{0,\, q_t + \gamma(\mathrm{err}_t - \alpha_{\mathrm{gate}})\},
\label{eq:aci}
\end{equation}
where $\mathrm{err}_t = \mathbf{1}\{s^{\mathrm{gate}}_t > q_t\}$ flags a miscoverage, the step size $\gamma$ sets how fast the radius moves, and the projection keeps the radius nonnegative, matching the relative error it bounds.
The radius $q_1$ is initialized at the empirical $(1-\alpha_{\mathrm{gate}})$ quantile of an in-control calibration window, the first stretch of scores collected before the gate begins admitting.
For a candidate action $a$, the DT's calibrated forward model (step {\Large \ding{177}}) predicts the post-action KPM vector $\hat{\mathbf{y}}_t(a)$, around which a calibrated prediction set is constructed with radius $q_t$,
\begin{equation}
C_t(a) = \big\{\mathbf{y} : \varrho\big(\mathbf{y},\hat{\mathbf{y}}_t(a)\big) \le q_t \big\},
\label{eq:confset}
\end{equation}
where $\varrho$ carries the weights of~\eqref{eq:global_dev} but divides each per-KPM residual by the prediction $\max(|\hat{y}^{(m)}_t(a)|,\epsilon_0)$ rather than by the realization, so the gate evaluates the set before the outcome arrives.
The gate then scores the action that the physical network actually ran, the admitted $a_t$ or the null action that leaves the network untouched, as $s^{\mathrm{gate}}_t=\varrho(\mathbf{x}^{\mathrm{post}}_t,\hat{\mathbf{y}}_t(a_t))$ on its realized post-action KPMs.
Eq.~\eqref{eq:aci} therefore compares that deviation against the same radius $q_t$ that admitted the action, before advancing the radius.
Scores are truncated at the declared range $B$, the largest deviation an operator asks the gate to distinguish, so $s^{\mathrm{gate}}_t$ stays in $[0,B]$.
The forward model predicts unchanged KPMs for an action that relocates no load, such as sleeping an already-idle cell.
A sleep action that does relocate load invokes a closed-form load-transfer model, which concentrates the traffic of a slept cell onto its least-loaded surviving neighbor, inflates that traffic by a factor for the lower modulation and coding scheme transferred users fall back to, and lowers their SINR by a fixed margin.

Running each what-if evaluation through this calibrated forward model rather than through a fresh simulation run preserves the gate's false-approval guarantee, since the gate bound derived in Section~\ref{sec:theory} is distribution-free and, for bounded scores, does not depend on the accuracy of that model.
The gate admits $a$ if and only if the entire prediction set lies inside the safe region, $C_t(a) \subseteq \mathcal{S}$.
Because an E2 control acknowledgment can be encoded inconsistently across independent implementations, an admitted action is confirmed at the RAN scheduler rather than at the acknowledgment.
Theoretical analysis below shows that under this pairing every false approval is a miscoverage event, so the long-run false-approval rate per decision step is bounded by the gate's miscoverage rate, which for the fixed-$\gamma$ update stays at most $\alpha_{\mathrm{gate}}+O(B/\gamma T)$ under arbitrary drift.

%% file: sections/theory.tex
\section{Theoretical Analysis}
\label{sec:theory}
Here, we prove the guarantees that keep each trust decision of \algname within its prescribed error budget, starting from two assumptions.
In operational terms, Theorem~\ref{thm:rls} bounds how closely the calibrator tracks a drifting deployment, and Theorem~\ref{thm:arl} bounds both how often resynchronization fires without cause and how long a genuine drift stays undetected.
Theorem~\ref{thm:gate} then bounds how often the gate admits an unsafe xApp action, and Corollary~\ref{cor:duty} combines these bounds into the fraction of time the DT runs degraded.

\begin{assumption}[Tracking regularity]
\label{as:track}
Let $\boldsymbol{\theta}^\star_t$ denote the instantaneous optimal calibration coefficients for the per-KPM RLS calibrator of Section~\ref{sec:rls} and $v_t=x_t-\mathbf{u}_t^{T}\boldsymbol{\theta}^\star_t$ the residual noise.
Then (i) within any interval between two consecutive alarms the input covariance $\mathbf{R}=\mathbb{E}[\mathbf{u}_t\mathbf{u}_t^{T}]$ is constant and positive definite with $\mathbf{R}\succeq\rho_0\mathbf{I}$ for some $\rho_0>0$, (ii) the optimal coefficients follow a random walk with independent zero-mean increments and $\mathbb{E}\|\boldsymbol{\theta}^\star_t-\boldsymbol{\theta}^\star_{t-1}\|^2\le\Delta^2$, (iii) the residual satisfies $\mathbb{E}[v_t\mid\mathcal{F}_{t-1},\mathbf{u}_t]=0$ with variance $\sigma^2$, where $\mathcal{F}_{t}$ is the natural filtration of the KPM stream, and (iv) the coefficient error $\boldsymbol{\theta}_{t-1}-\boldsymbol{\theta}^\star_{t-1}$ is independent of the current regressor $\mathbf{u}_t$, which is the independence approximation standard in tracking analysis of adaptive filters.
\end{assumption}

In plain terms, Assumption~\ref{as:track} asks for a few mild properties of the KPM stream: the deployment drifts by a bounded amount per step, the measurement noise averages to zero, and the calibrator inputs vary enough to pin down the coefficients.

\begin{assumption}[In-control sub-Gaussian increment]
\label{as:incontrol}
Under the in-control regime $H_0$ that holds between two genuine changes, the standardized deviation $z_t$ of Section~\ref{sec:edetector} is conditionally $1$-sub-Gaussian with non-positive mean, $\mathbb{E}_{H_0}[\exp(\xi z_t)\mid\mathcal{F}_{t-1}]\le \exp(\xi^2/2)$ for all $\xi\ge 0$, the standard in-control condition of sequential change detection~\cite{shin2022edetectors, ramdas2023gametheoretic}.
\end{assumption}

The condition says only that large upward jumps of the standardized deviation stay rare while the DT matches the deployment, and Section~\ref{sec:results} checks it on a public KPM trace.


\subsection{Tracking Floors of the Adaptive DT Calibrator}
\label{sec:adaptivefloor}

\begin{theorem}[Tracking-error floor and optimal $\lambda$]
\label{thm:rls}
Under Assumption~\ref{as:track}, the steady-state tracking error of the RLS calibrator with forgetting factor $\lambda\in(0,1)$ within a sufficiently long regime satisfies
\begin{equation}
\mathbb{E}\|\boldsymbol{\theta}_t-\boldsymbol{\theta}^\star_t\|^2 \;\lesssim\; C_1(1-\lambda)\sigma^2 + \frac{C_2\,\Delta^2}{1-\lambda},
\label{eq:rlsbound}
\end{equation}
where $C_1,C_2>0$ depend on $\rho_0$ and the filter dimension, and minimizing the right-hand side gives the optimal forgetting factor $1-\lambda^\star=\sqrt{C_2/C_1}\,\Delta/\sigma$ whenever $\sqrt{C_2/C_1}\,\Delta/\sigma<1$, at which the bound equals $2\sqrt{C_1 C_2}\,\sigma\,\Delta$.
\end{theorem}

The first term of Eq.~\eqref{eq:rlsbound} grows when the calibrator forgets too fast and chases noise, and the second grows when it forgets too slowly and lags the moving deployment.
Their optimized sum is the calibrator's tracking floor, the smallest error the analysis certifies rather than a limit on achievable error.
Thus the best forgetting factor scales with the ratio $\Delta/\sigma$ of drift to noise, so the calibrator forgets faster when the deployment moves quickly and more slowly when measurements are noisy.

\begin{proof}[Proof sketch]
The tracking error $\tilde{\boldsymbol{\theta}}_t=\boldsymbol{\theta}_t-\boldsymbol{\theta}^\star_t$ obeys the recursion $\tilde{\boldsymbol{\theta}}_t=(\mathbf{I}-\mathbf{k}_t\mathbf{u}_t^{T})\tilde{\boldsymbol{\theta}}_{t-1}-(\mathbf{I}-\mathbf{k}_t\mathbf{u}_t^{T})(\boldsymbol{\theta}^\star_t-\boldsymbol{\theta}^\star_{t-1})+\mathbf{k}_t v_t$, and taking expectations under Assumption~\ref{as:track}(i) and~(iv) bounds the contraction of $\mathbf{I}-\mathbf{k}_t\mathbf{u}_t^{T}$ and the gain covariance $\mathbb{E}[\mathbf{k}_t\mathbf{k}_t^{T}]$.
The steady-state error covariance then separates into the lag and noise terms of Eq.~\eqref{eq:rlsbound}, adapting standard RLS tracking analysis to the calibration setting, and the arithmetic-geometric-mean inequality over $\lambda$ yields the optimum.
\end{proof}

The fixed-share aggregator of Section~\ref{sec:rls} inherits the standard tracking-regret bound~\cite{orabona2026online}: against any reference strategy that switches experts $k$ times over $T$ rounds, its average loss exceeds the per-block best expert's by a slack that vanishes as $T$ grows.
Letting that reference switch at the resynchronization boundaries of Section~\ref{sec:edetector} makes $k$ at most the number of epochs, so tracking and delay bounds compose along the same regime boundaries.

\subsection{Resynchronization Guarantee}

\begin{theorem}[False-alarm control under any inspection schedule and detection delay]
\label{thm:arl}
Consider the mixture e-process $E_t$ of Eq.~\eqref{eq:eprocess} over the one-sided grid $\{\delta_k>0\}$ with the alarm time $\tau=\inf\{t>\nu_0 : E_t \ge 1/\alpha_{\mathrm{res}}\}$, and let Assumption~\ref{as:incontrol} hold.
(i) The process $\{E_t\}$ is a nonnegative supermartingale with $\mathbb{E}_{H_0}[E_t]\le 1$, so Ville's inequality~\cite{ramdas2023gametheoretic} gives $\mathbb{P}_{H_0}(\tau<\infty)\le \alpha_{\mathrm{res}}$ for any inspection schedule that depends only on the data observed so far.
(ii) If a change at time $\nu$ coincides with the last reset, $\nu=\nu_0$, so that $E_\nu=1$, and the post-change increments are i.i.d.\ with finite second moment and $\mathbb{E}[\delta_{k^\star} z_t-\delta_{k^\star}^2/2]\ge I>0$ for some grid shift $\delta_{k^\star}$, then
\begin{equation}
\mathbb{E}[\tau-\nu\mid\tau\ge\nu]\;\le\;\frac{\ln(1/\alpha_{\mathrm{res}})+c_\pi+\kappa}{I},
\label{eq:delay}
\end{equation}
with $c_\pi=\ln(1/\pi_{k^\star})$ the cost of the mixture weight placed on that component and $\kappa$ the expected overshoot, the amount by which the matched component's log-likelihood walk passes the alarm boundary at the crossing step.
The overshoot depends on the post-change increment law and the shift size, yet stays bounded as the threshold grows.
\end{theorem}

In other words, within each epoch the chance of ever raising a spurious resynchronization stays at most $\alpha_{\mathrm{res}}$, however often the supervisor inspects the deviation stream.
Thus a genuine drift is caught within a delay that grows only logarithmically in $1/\alpha_{\mathrm{res}}$ and shrinks as $I$ grows, where $I$ measures how much evidence each post-change KPM sample adds.
A larger drift makes $I$ larger, so larger drifts are caught faster.

\begin{proof}[Proof sketch]
Under Assumption~\ref{as:incontrol} each factor $\exp(\delta_k z_t-\delta_k^2/2)$ has conditional mean at most one, so every component product is a nonnegative supermartingale and the finite mixture preserves this property, whereupon Ville's maximal inequality yields part (i).
After a change at the reset, the i.i.d.\ post-change increments make the matched component's log-likelihood a random walk with drift at least $I$, which $\ln E_t$ dominates up to the additive cost $c_\pi$, and Wald's identity with the overshoot term $\kappa$ bounds the expected first passage by Eq.~\eqref{eq:delay}, adapting the e-detector delay analysis of~\cite{shin2022edetectors} to this mixture.
\end{proof}

\subsection{Conformal Fidelity Gate Guarantee}

\begin{theorem}[Gate coverage and false-approval control]
\label{thm:gate}
Let the online quantile of the gate of Section~\ref{sec:gate} follow the projected ACI update Eq.~\eqref{eq:aci} with step size $\gamma>0$ and initial radius $q_1\in[0,B]$ on any possibly adversarially drifting gate-score sequence whose values are bounded, $s^{\mathrm{gate}}_t\in[0,B]$, and let the safe region $\mathcal{S}$ be fixed.
(i) The empirical miscoverage satisfies the one-sided bound
\begin{equation}
\tfrac{1}{T}\textstyle\sum_{t=1}^{T}\mathrm{err}_t-\alpha_{\mathrm{gate}}\;\le\;\frac{B+\gamma}{\gamma T}\;\xrightarrow[T\to\infty]{}\;0,
\label{eq:aci_cov}
\end{equation}
with no exchangeability assumption on the score sequence, that is, the scores may drift over time and their order may carry information.
(ii) For every action $a$, the containment $C_t(a)\subseteq\mathcal{S}$ is nonincreasing in the radius $q_t$, so the false-approval loss $\ell_t(a)=\mathbf{1}\{\text{admit }a\wedge\text{true post-action KPM}\notin\mathcal{S}\}$ is a monotone nonincreasing function of $q_t$.
(iii) Suppose the score $s^{\mathrm{gate}}_t$ is the deviation realized under the action applied at step $t$, and suppose Eq.~\eqref{eq:aci} compares that score against the same radius that admitted the action, as Algorithm~\ref{alg:OpenTwin} pairs them.
Then an applied action whose true post-action KPM falls outside $\mathcal{S}$ realizes a score above that radius, so every false approval is a miscoverage event and, with $\ell_t$ the false-approval loss of that action,
\begin{equation}
\tfrac{1}{T}\textstyle\sum_{t=1}^{T}\ell_t\;\le\;\tfrac{1}{T}\textstyle\sum_{t=1}^{T}\mathrm{err}_t\;\le\;\alpha_{\mathrm{gate}}+O\!\Big(\tfrac{B}{\gamma T}\Big).
\label{eq:crc}
\end{equation}
\end{theorem}

In other words, updating the radius with a fixed step size keeps the long-run rate of unsafe admissions at or below the prescribed budget, so an operator who sets $\alpha_{\mathrm{gate}}=0.1$ accepts at most roughly one unsafe admission per ten decision steps in the long run.
The dynamically tuned variant~\cite{gibbs2022conformal} removes the need to choose a step size but does not inherit this bound.

\begin{proof}[Proof sketch]
Part (i) follows by telescoping Eq.~\eqref{eq:aci}, since bounded scores keep $q_{T+1}-q_1\le B+\gamma$ and the projection mass only lowers the averaged miscoverage, preserving the upper direction of the bound of~\cite{gibbs2021adaptive}.
Part (ii) is monotone-loss hypothesis of conformal risk control~\cite{angelopoulos2022conformal}, and part (iii) holds because an admitted action whose true outcome falls outside $\mathcal{S}$ realizes a score above the admitting radius, making every false approval a miscoverage event that part (i) bounds.
\end{proof}

\subsection{A Duty-Cycle Guarantee on Fidelity}
Duty cycle here denotes the fraction of KPM samples on which the DT runs degraded rather than tracking the deployment, and the term is unrelated to radio transmission.

\begin{corollary}[Fidelity duty cycle]
\label{cor:duty}
Under the assumptions of Theorems~\ref{thm:rls} and~\ref{thm:arl}, the aggregator's average loss within a between-drift regime exceeds the better expert's per-block steady-state loss by at most the vanishing fixed-share slack~\cite{orabona2026online}, and Theorem~\ref{thm:rls} gives that steady-state level for the RLS expert in the pool.
This loss, the calibrator's absolute error truncated at one, rises with the calibrated deviation of Section~\ref{sec:dev-score}, so below the truncation point a bound on either quantity carries to the other.
A genuine drift then degrades the DT for at most $\bar{D}+T_{\mathrm{resync}}$ samples in expectation, where $\bar{D}$ is the expected detection delay and $T_{\mathrm{resync}}$ is the regeneration window that the supervisor reserves for re-inverting configuration parameters, resetting the calibrator, and letting the deviation score settle again.
Theorem~\ref{thm:arl} bounds $\bar{D}$ by $D=(\ln(1/\alpha_{\mathrm{res}})+c_\pi+\kappa)/I$ for an onset that coincides with a reset.
A drift that begins at a mid-epoch time $\nu$ adds the term $\ln(1/E^{(k^\star)}_{\nu})$, the additional evidence the matched component $E^{(k^\star)}_t$ inside Eq.~\eqref{eq:eprocess} must accumulate because its value at time $\nu$ may sit below one.
Now suppose that genuine regime changes arrive on average every $T_{\mathrm{drift}}$ samples and that successive between-alarm epochs restart independently.
Each epoch ends in a spurious resynchronization with probability at most $\alpha_{\mathrm{res}}$, and each spurious alarm costs one regeneration window.
The long-run degraded fraction then stays at most $\min\{1,(\bar{D}+T_{\mathrm{resync}}+\tfrac{\alpha_{\mathrm{res}}}{1-\alpha_{\mathrm{res}}}T_{\mathrm{resync}})/T_{\mathrm{drift}}\}$.
\end{corollary}

Thus, \algname accepts a bounded fraction of degraded time in return for automatic recovery from genuine drifts.
Outside those windows, the calibrator tracks the deployment at its per-regime floor.

%% file: sections/eval.tex
\section{Evaluation Results}
\label{sec:results}
\subsection{Simulation-Based Experimental Setup}
\label{sub5.1}

\myparatight{Platform}
The physical network runs a complete 5G software stack, OpenAirInterface (OAI)~\cite{kaltenberger2020oai}, with a gNB and multiple nrUEs in RFsimulator mode.
The full Radio Resource Control (RRC), Packet Data Convergence Protocol (PDCP), Radio Link Control (RLC), Medium Access Control (MAC), and physical-layer pipeline therefore executes unchanged, and only the radio channel is emulated in software.
An integrated E2 agent compiled to E2 Application Protocol (E2AP) with E2 Service Model for Key Performance Measurement (E2SM-KPM) reports standardized KPMs, and a single near-RT RIC instance, FlexRIC~\cite{schmidt2021flexric}, connects both sides.
The backbone simulator of the DT is the packet-level ns-O-RAN-flexRIC~\cite{ns-O-RAN-flexric},
and fully simulated setup replaces OAI physical side with an independently seeded run of the same simulator.

\myparatight{Inference details and metrics}
Fields carried in the E2SM-KPM report are read directly, while radio, mobility, and buffering parameters are predicted per feature window.
The predictor trains on a 300-run configuration sweep under a group-by-run split, which keeps every window of a given run entirely in training or entirely in test, and uses 300 trees of depth four at learning rate 0.1, a 30\% held-out run fraction, and five-seed averaging.
A pair denotes one physical run and the DT run that mirrors it.
We report MAPE and, for categorical parameters, the exact-match rate.
For trust decisions we report the gate's coverage, the fraction of steps whose true outcome falls inside the predicted set, and its false-approval rate against the budget $\alpha_{\mathrm{gate}}$, with the detector's false-alarm rate and detection delay against $\alpha_{\mathrm{res}}$.
Energy is scored against an always-on baseline with a linear base-station power model~\cite{lopezperez2022energy}, in which an idle macro site draws approximately $780$\,W and a sleeping one approximately $450$\,W.

\subsection{DT Construction Fidelity}
\label{sub:certification}
\begin{figure}[t]
\centering
\includegraphics[width=0.78\columnwidth]{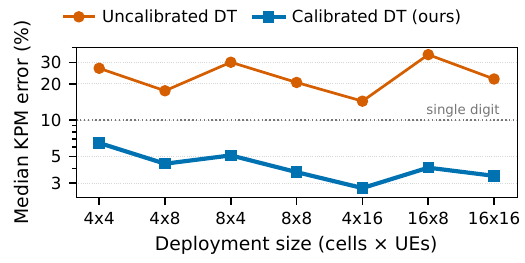}
\caption{DT construction fidelity across seven deployment sizes, as medians over twelve to fourteen paired runs per size.}
\label{fig:fidelity}
\vspace{-.08in}
\end{figure}

\myparatight{DT fidelity}
Figure~\ref{fig:fidelity} compares the uncalibrated and the calibrated DT across seven deployment sizes, from 4 cells with 4 UEs up to 16 cells with 16 UEs.
Inferred configuration parameters alone leave 14.4\% to 34.9\% median KPM error, and online calibration reduces that error to 2.7\% to 6.5\%.
The calibrated DT shows no systematic degradation as the deployment grows sixteenfold.
Its error also stays at or below the simulator's own run-to-run variability at five of the seven sizes and within a point of it at the other two.
That variability is the gap between two independently seeded runs of the same configuration, so a calibrated DT matches the physical network about as well as the network matches itself across runs.
Calibration rather than inversion therefore delivers the single-digit fidelity the conformal fidelity gate later needs, using only the KPM stream the deployment already exports.

\myparatight{Telemetry certification by re-simulation}
Re-simulation drives the uncalibrated DT's monitored telemetry deviation from 18.9\% at the initial inversion down to 10.3\% on held-out targets, and the learned initialization is essential, since the same search started from the midpoint of parameter ranges reaches only 40.5\%.
The central finding is that telemetry improvement does not imply configuration recovery.
Even as monitored telemetry improves, only about four of the eleven continuous parameters move toward their true values while the rest move away.
Many different configurations therefore produce the same observable KPMs.
A flexible correction term can absorb errors that belong to the parameters themselves~\cite{xie2020projected}, so a low KPM error says nothing about whether the configuration is correct.
\algname places its trust in the guarantee layer rather than in configuration recovery.

\subsection{Online DT Calibration}
The RLS expert uses forgetting factor $\lambda=0.99$, window $W=5$, lag order $L=3$, and regularization $\rho=10^{-2}$, retuned to $\lambda=0.90$, $W=10$, and $\rho=10^{-3}$ for the cross-implementation experiments of Section~\ref{sub:crossimpl}.
The fixed-share aggregator uses $\alpha_{\text{sh}}=0.05$ with Hedge learning rate $\eta=2$.
Across 12 paired runs, six stationary and six with injected drift, Fig.~\ref{fig:corrbars} reports 5.89\% MAPE for the adaptive calibrator against 8.72\% for the Kalman filter and 10.59\% for RLS, the two experts it blends.
The adaptive calibrator ties RLS in the stationary condition that favors RLS and beats both experts under drift.
The design target, however, is stable DT fidelity across conditions rather than the lowest overall error.
The adaptive calibrator narrows the spread of MAPE across conditions to about 2.5 points, against 12.3 points for RLS.
Over all runs it also stays within 0.60 MAPE of the best expert chosen in hindsight for each condition, the tracking behavior the fixed-share bound predicts.
A held-out validation on 78 pairs never used for tuning confirms that advantage over Kalman, RLS, exponentially weighted moving average (EWMA), and a persistence baseline that repeats the last observed value, at $p<0.0001$.
A synthetic scalar instance separately validates the tracking floor of Theorem~\ref{thm:rls}, where the measured error traces the predicted convex curve in $1-\lambda$.
The strongest sequence model, Chronos-Bolt~\cite{ansari2024chronos} applied without any training on this data, matches the calibrated error within noise, yet it needs roughly 18\,ms per update, far beyond the near-RT latency budget, against roughly 32\,$\mu$s per KPM.

\begin{figure}[t]
\centering
\includegraphics[width=0.78\columnwidth]{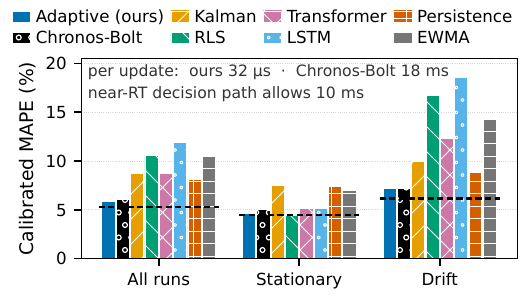}
\caption{Calibrated MAPE per calibrator on 12 paired runs, over all runs and split by condition, where dashed ticks mark the better expert for each condition, chosen after the fact.}
\label{fig:corrbars}
\vspace{-.08in}
\end{figure}

\subsection{Freshness Guarantee Validation}
\label{sub:guarantee}
\myparatight{Drift detection under any inspection schedule}
The drift detector standardizes the deviation stream with robust in-control statistics $\mu_0=2.813$ and $\sigma_0=1.434$, averages its evidence over a grid of candidate upward shifts, and runs at a resynchronization budget $\alpha_{\mathrm{res}}=1/1000$.
Table~\ref{tab:peeking} re-runs every detector at a common 0.05 false-alarm budget, where our detector stays inside that budget under every inspection schedule and never exceeds 0.023.
A fixed-time $z$ test instead alarms with probability above 0.5 when inspected at every step, so continuous inspection is statistically free only for the anytime-valid detectors.
The sub-Gaussian bound of Assumption~\ref{as:incontrol} in Section~\ref{sec:theory} also holds on a public KPM trace from a separate OAI deployment, and across two million simulated in-control steps the detector raises no false alarm while its detection delay decays as $1/\delta^2$, from 68.9 samples at $\delta=0.5\sigma$ to 2.5 at $3\sigma$.
Figure~\ref{fig:detdelay} plots that delay against the analytic floor $\ln(1/\alpha_{\mathrm{res}})/I$, together with the realized duty cycle of Corollary~\ref{cor:duty}.
That fraction falls from 1.00 to 0.16 as the drift grows while both periodic policies stay flat, so detection-triggered resynchronization costs less than either periodic schedule once the drift exceeds roughly $1.5\sigma$.

\begin{figure}[t]
\centering
\includegraphics[width=0.78\columnwidth]{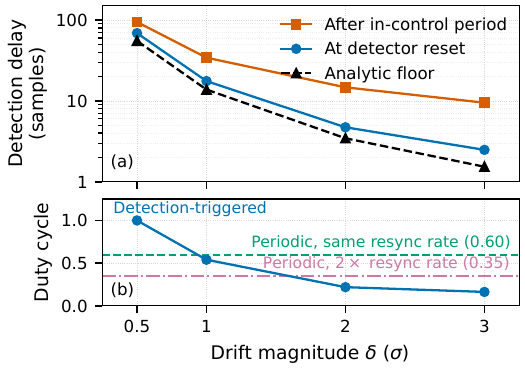}
\caption{Detection delay against the analytic floor $\ln(1/\alpha_{\mathrm{res}})/I$, and duty cycle against periodic policies ($T_{\mathrm{resync}}=4$, $T_{\mathrm{drift}}=40$ samples, 2000 repetitions).}
\label{fig:detdelay}
\vspace{-.08in}
\end{figure}

\begin{table}[t]
\centering
\caption{Probability that each detector ever alarms within 5000 samples with no drift present, across inspection schedules.
Conf.\ seq.\ denotes the time-uniform confidence sequence of~\cite{howard2021timeuniform}, Geometric $\times 2$ inspects at doubling intervals, and Terminal only inspects once at the end.}
\label{tab:peeking}
\renewcommand{\arraystretch}{1.0}\footnotesize
\begin{tabular}{lccc}
\toprule
\textbf{Inspection} & \textbf{Ours} & \textbf{Conf.\ seq.}~\cite{howard2021timeuniform} & \textbf{Fixed-time $z$} \\
\midrule
Every step & 0.023 & 0.035 & 0.512 \\
Every 10 & 0.009 & 0.031 & 0.403 \\
Every 100 & 0.003 & 0.021 & 0.272 \\
Geometric $\times 2$ & 0.011 & 0.010 & 0.330 \\
Terminal only & 0.000 & 0.001 & 0.048 \\
\bottomrule
\end{tabular}
\end{table}

\myparatight{Conformal fidelity gate coverage under drift}
The safe region $\mathcal{S}$ caps cell load at $50\%$ and requires SINR above $-12$\,dB, and the budget $\alpha_{\mathrm{gate}}$ defaults to $0.10$ and is swept over $\{0.02,0.05,0.10,0.20,0.30\}$.
Every gate number below comes from the fixed-step update covered by Theorem~\ref{thm:gate}.
The adaptive radius agrees with the reference MAPIE implementation~\cite{taquet2022mapie} within run-to-run noise.
On scores larger than any seen during calibration it still holds coverage of 0.878 against a 0.90 target, while a fixed radius, calibrated once and never updated, stops at 0.746.
On the fully simulated closed loop, in-control coverage stays within four points of the $1-\alpha_{\mathrm{gate}}$ target in Fig.~\ref{fig:gatecov}, while the drifting and bursty-traffic regimes reach only 0.843 and 0.852 against the 0.90 target.
These runs restart the adaptive update and last only 27 to 35 decision steps, so the radius grows by at most 1.75 deviation points where post-onset scores need 3.5 to 9.6, and the shortfall reflects the short calibration period rather than the gate itself.
The false-approval rate stays low despite that shortfall, because at every budget at least 69.6\% of the steps whose true outcome falls outside the predicted set are steps the gate already refused, and every recorded decision step satisfies the per-step inequality of Theorem~\ref{thm:gate}(iii).
An always-admit xApp instead violates the safe region at rates of 0.346 in control and 0.111 under drift.
The gate therefore fails by refusing safe actions rather than by admitting unsafe ones.

\begin{figure}[t]
\centering
\includegraphics[width=0.78\columnwidth]{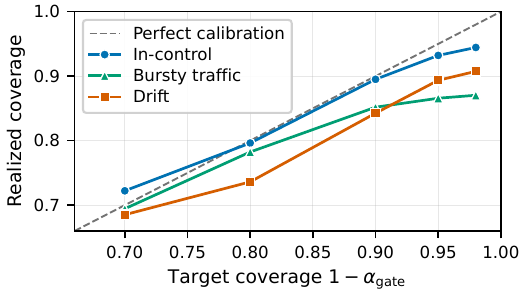}
\caption{Gate coverage against the $1-\alpha_{\mathrm{gate}}$ target, with the drifting and bursty-traffic series calibrated pre-onset and evaluated post-onset (six to eight seeds).}
\label{fig:gatecov}
\vspace{-.08in}
\end{figure}

\subsection{DT-Validated Energy-Saving xApp}
The energy-saving xApp monitors per-UE SINR over E2 and hands a UE to the neighboring cell with the best smoothed SINR when its serving SINR degrades.
It puts a cell to sleep once its last UE leaves and reactivates sleeping cells periodically.
Sweeping the violation budget traces the trade-off curve of Fig.~\ref{fig:opcurve}, where energy saving and admit rate both rise with $\alpha_{\mathrm{gate}}$ and the false-approval rate stays at zero.
At the loosest budget the gate saves 6.12\% energy on the drift scenario, against 11.1\% for an unguarded always-admit xApp that pays the violation rates reported above.
The sweep also exposes a limit on how small a budget the gate can certify.
The gate calibrates on $n_{\mathrm{cal}}$ in-control segments and resolves no budget finer than $1/(n_{\mathrm{cal}}+1)$, so smaller budgets carry measured rather than certified rates.
Lengthening each run to 60\,s gives $n_{\mathrm{cal}}=39$ and a limit of 0.025.
At a 0.05 budget the worst drift scenario then realizes a false-approval rate of 0.0016, more than an order of magnitude below the budget.
A separate set of 210 drift pairs with 60\,s runs measures how much safe action the gate still refuses, using an oracle that sees the true outcome of every action.
At $\alpha_{\mathrm{gate}}=0.30$ the gate admits 0.470 of the actions the oracle marks safe and captures 0.401 of the achievable energy.
Of the actions it admits, 92.3\% prove safe.
Its refusals are selective, since it admits 0.265 of the actions proposed by a safe policy but only 0.019 of the actions proposed by a deliberately unsafe cell-sleep policy.
The DT is conservative in the same direction, since on bursty traffic it predicts a 14.1\% saving for the always-admit controller while the physical network delivers 16.6\%.
\algname therefore gives up some energy rather than claiming savings the physical network cannot deliver.

\begin{figure}[t]
\centering
\includegraphics[width=0.78\columnwidth]{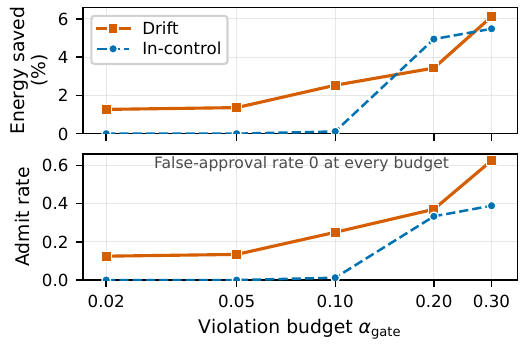}
\caption{Energy saved and admit rate of the conformal fidelity gate across violation budgets, under the prescribed safe region.}
\label{fig:opcurve}
\vspace{-.08in}
\end{figure}

\subsection{Scalability and Overhead}
A fit over 470 runs spanning 1 to 32 cells and 1 to 64 UEs yields one cost model for the real-time factor (RTF), the wall-clock seconds the backbone simulator spends per simulated second, $\mathrm{RTF} = K_{\mathrm{machine}} \cdot T_{\mathrm{traffic}} \cdot \mathrm{cells}^{1.19} \cdot \mathrm{UEs}^{0.95}$ with $R^2=0.949$.
Figure~\ref{fig:rtf} shows measurements from four machines matching that model at a median measured-to-predicted ratio of 1.02 over all 507 recorded runs.
A single DT process stays under 2.5\,GB at every size, so compute rather than memory sets the scaling limit.
One iteration of the candidate loop of Algorithm~\ref{alg:OpenTwin} costs approximately 7.5\,$\mu$s at the median over 62,370 gate decisions.
A complete online update, including deviation scoring and the detector update, ranges from approximately 91\,$\mu$s for the per-KPM RLS-only calibrator to approximately 1.2\,ms for the adaptive calibrator at twenty monitored KPMs, well inside the near-RT loop's 10\,ms bound.
Running the DT forward is the only expensive operation, at approximately 123\,s per eight-second closed-loop segment at four cells and six UEs, and it runs only on resynchronization rather than inside the per-action loop.
The cost model also predicts the length of the regeneration window that Corollary~\ref{cor:duty} needs, from a machine constant and a topology the operator already knows.

\begin{figure}[t]
\centering
\includegraphics[width=0.78\columnwidth]{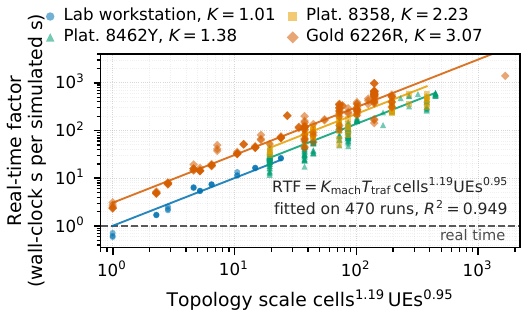}
\caption{Measured real-time factor of the 507 recorded scalability runs against the cost model fitted on 470 runs.}
\label{fig:rtf}
\vspace{-.08in}
\end{figure}


%% file: sections/exp.tex
\section{Testbed Experiments and Validation}
\label{sec:testbed}

\subsection{Testbed Setup}
\label{sub:testbed}
\begin{figure}[t]
\centering
\subfloat[]{\includegraphics[height=1in]{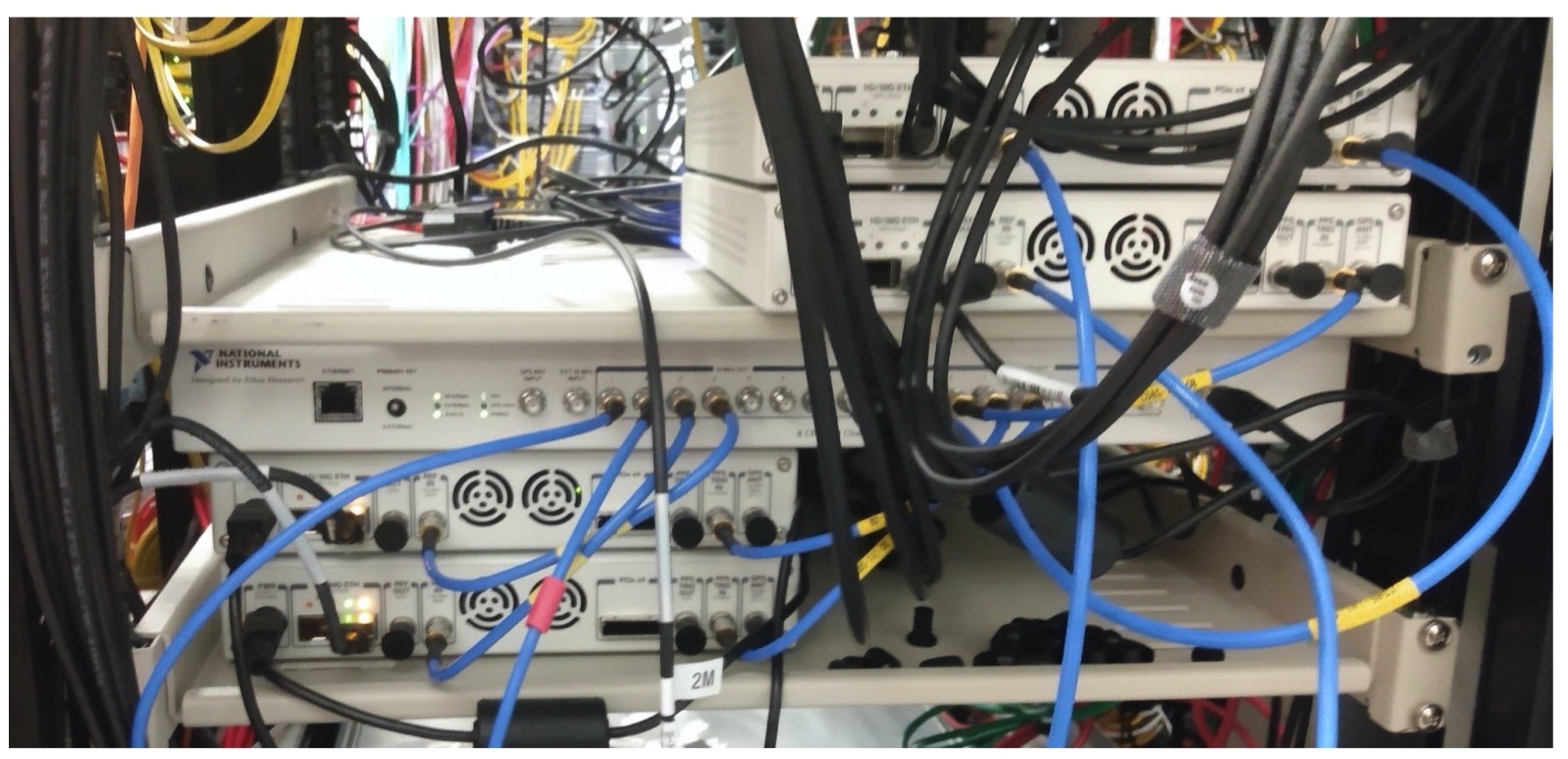}}\hspace{4pt}
\subfloat[]{\includegraphics[height=1in]{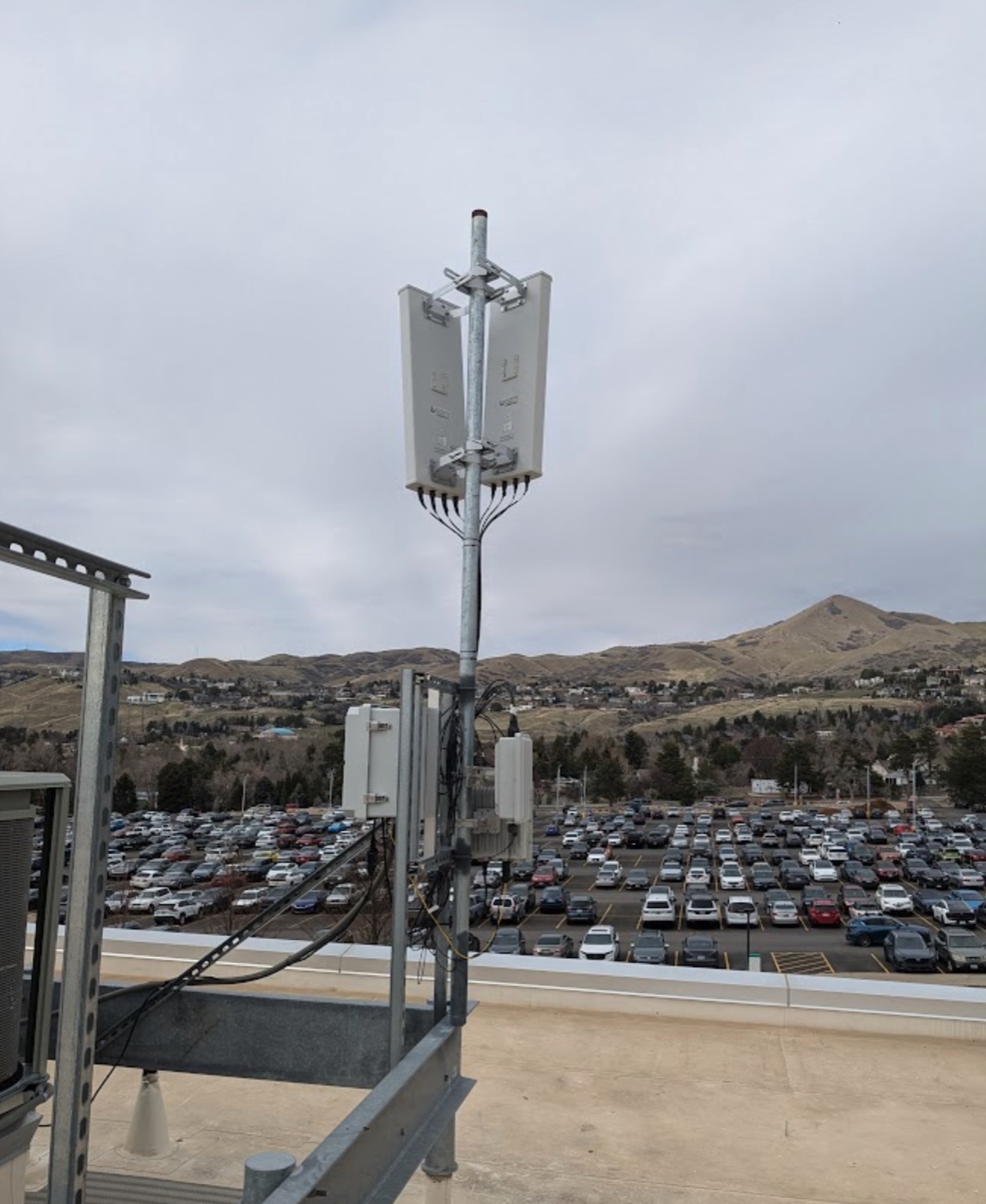}}
\caption{The POWDER testbed of the radio-hardware study: (a) a paired-radio workbench of the kind this study uses, where USRP X310 radios connect over attenuated coaxial links; (b) a rooftop radio unit of the same testbed.
Photos: POWDER testbed, University of Utah.}
\label{fig:powder}
\vspace{-.08in}
\end{figure}
The radio-hardware testbed runs an OAI gNB on a USRP X310 software-defined radio on the POWDER platform~\cite{breen2021powder}, and the gNB serves a commercial Quectel RM520N-GL module through a programmable attenuator matrix, as Fig.~\ref{fig:powder}(a) shows.
Real 5G signals therefore cross a conducted RF path whose channel stays scripted and replayable, and the commanded attenuation serves as physical ground truth for scoring inversion.
Radio hardware limits what an experiment can vary, since the installed radios and topology leave little beyond the commanded attenuation.
The software-radio setup therefore complements the hardware one by running the same OAI stack under RFsimulator, where the entire protocol stack executes unchanged and only the radio channel is replaced by a software link.
It remains a real testbed, since everything above the radio channel is the deployed implementation, and it can sweep deployment sizes, traffic patterns, and configuration parameters far beyond what fixed hardware allows.
The multi-vendor configuration adds a second vendor to that OAI stack, attaching an srsRAN Project gNB~\cite{srsran2025project} that speaks E2AP v3.00 to the same v1.01 RIC, and seven E2 nodes then stay registered at once with live KPM subscriptions.
The srsRAN node runs its complete PDCP, RLC, and MAC stack and its real scheduler over a dummy radio unit that terminates the radio interface in software.

\subsection{Closed-Loop Fidelity on the OAI Software Stack}
\label{sub:crossimpl}

This study runs the entire loop against a physical KPM stream from the OAI stack while the DT keeps its own backbone simulator, so the DT reproduces a deployment that shares none of its code.
Twenty collections of 600\,s each, with KPMs reported at 1\,Hz, drive one closed-loop evaluation apiece.
Twelve run in control at 15 to 60\,Mbps downlink with one to three UEs, and the other eight carry a traffic change at a recorded mid-run time.
Table~\ref{tab:crossimpl} reports four KPMs because the two implementations support different subsets of the E2SM-KPM catalog, and those subsets overlap on exactly these four downlink counters, itself a cross-implementation finding.
The OAI stream exposes 27 of the 128 inversion features, namely those that these four counters can supply, and those 27 suffice to initialize the DT, drive the calibrator, and carry every trust decision.
E2SM-KPM standardizes how a report is encoded but leaves each implementation to decide how it computes each counter, so two compliant stacks are expected to differ by a gap that no amount of modeling can remove.
Even across that gap, the calibrated loop keeps per-KPM error between 0.2\% and 4.9\% on the in-control runs.
One drift run leaves a physical cell idle, and its near-zero relative-error denominator inflates the standard deviations in the drift column.
The detector flags six drift events one KPM indication after onset, five of them scripted changes and one an unscripted traffic-source failure verified in the raw per-UE counters.
Combining that delay with a nine-sample regeneration window through Corollary~\ref{cor:duty} leaves the DT degraded for 0.16 to 0.19 of the time.
Three scripted changes raise no alarm.
One falls inside a resynchronization window that is already open, and the calibrator alone tracks the other two.
The conformal fidelity gate carries over with the same short-horizon coverage shortfall that Section~\ref{sub:guarantee} analyzes, reaching coverage of $0.852\pm0.149$ against a 0.90 target with a zero false-approval rate.
Configuration inference recovers the topology parameters exactly, while bandwidth and offered load stay limited by the four mapped KPMs.
The loop's trust therefore rests on measured agreement rather than on shared implementation.

\begin{table}[t]
\centering
\caption{Calibrated cross-implementation DT fidelity on the OAI-physical closed loop, as evaluation-window MAPE (\%), mean$\pm$std.
Pair type records whether both stacks report a KPM in the same unit or one is rescaled.}
\label{tab:crossimpl}
\renewcommand{\arraystretch}{1.0}\footnotesize
\begin{tabular}{llcc}
\toprule
\textbf{KPM} & \textbf{Pair type} & \textbf{In-control (12)} & \textbf{Drift (8)} \\
\midrule
Active-UE count & same unit & 0.23 $\pm$ 0.42 & 9.83 $\pm$ 27.15 \\
PRB utilization & same unit & 1.47 $\pm$ 2.04 & 20.26 $\pm$ 38.75 \\
Used-PRB volume & rescaled & 1.54 $\pm$ 2.09 & 13.99 $\pm$ 20.54 \\
User-plane volume & rescaled & 4.90 $\pm$ 10.07 & 19.06 $\pm$ 29.06 \\
\bottomrule
\end{tabular}
\end{table}

\subsection{Controlled RF Validation on Radio Hardware}
\label{sub:rf}
The full loop runs end to end on the radio hardware, where inversion estimates the applied attenuation with an error of 5.82\,dB.
The DT then evaluates candidate recovery actions, the gate screens them without using that estimate, and the approved action restores downlink throughput to approximately 119\,Mbps.
A separate session fixes the calibrated radius at $q_t=11.0\%$ and sets a 30\,Mbps throughput floor, and the gate then approves attenuation steps of 0 to 30\,dB and refuses 35 and 40\,dB, as in Fig.~\ref{fig:rfgate}.
Runs executed with the gate bypassed confirm that the refused 40\,dB step breaks the floor, while the 35\,dB step lands just above it and the DT under-predicts its throughput by approximately the calibrated radius.
No unsafe action is approved, and a two-round drift sequence exercises the full resynchronization cycle, with each shock detected two KPM indications after onset.
Five repetitions of the same channel replay agree to within a median of 0.5\% per KPM.
The gate therefore separates safe from unsafe attenuation steps on a physical radio path, extending the admission mechanism to radio hardware without redesign.

\begin{figure}[t]
\centering
\includegraphics[width=0.8\columnwidth]{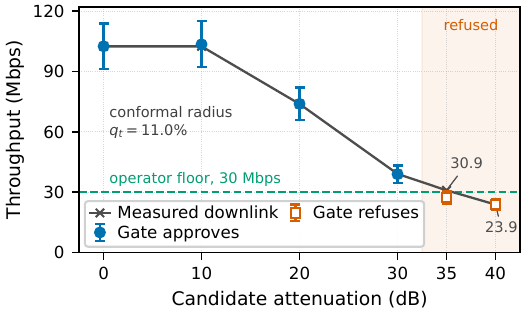}
\caption{Gate decisions on the radio-hardware testbed, where error bars show the DT's predicted throughput range at the calibrated radius.
The 35 and 40\,dB points were executed with the gate bypassed, outside the loop.}
\label{fig:rfgate}
\end{figure}

\subsection{Multi-Vendor Interoperability on a Shared RIC}
\label{sub:interop}

A translation layer rewrites the four message-encoding differences between E2AP v1.01 and v3.00 before setup succeeds.
We then verify actuation at the scheduler rather than at the acknowledgment, and the srsRAN scheduler allocates exactly the commanded quota of $\lfloor \zeta/100 \times 51 \rfloor$ physical resource blocks (PRBs) at every commanded ratio $\zeta$.
Acknowledgments cannot provide that verification, since every E2SM control handler in the evaluated OAI release returns success without invoking the MAC, RRC, or scheduler.
An xApp therefore cannot trust the interface across implementations, and guarantees grounded in KPMs remain the only auditable evidence that its action took effect.
Attaching the second implementation, together with the radio hardware, exposed the seven interoperability defects of Table~\ref{tab:defects}.
Five of them cannot appear in a single-vendor deployment, because a RIC co-developed with its own E2 agent hard-codes that agent's optional information element (IE) conventions and feature subset as assumptions.
These failures are also disproportionate, since a deviation from those assumptions crashes the shared control plane or silently inverts a control acknowledgment so the action applies while the requester is told it did not.
Single-vendor co-development therefore hides the failure modes an open control plane exposes, so exercising the loop against a stack it was not developed with is part of validating it.


\begin{table}[t]
\centering
\caption{Interoperability defects exposed by the multi-vendor and radio-hardware testbeds against the shared RIC.}
\label{tab:defects}
\renewcommand{\arraystretch}{1.0}\footnotesize\setlength{\tabcolsep}{3pt}
\begin{tabular}{p{0.40\columnwidth}p{0.54\columnwidth}}
\toprule
\textbf{Encoded assumption} & \textbf{Observed failure} \\
\midrule
D1: an optional IE is never present & every node-enumerating xApp aborts \\
D2: an optional IE is always present & the RIC crashes at subscription, disconnecting all nodes \\
D3: every advertised RAN function is known & assertion abort instead of the rejected-function list \\
D4: the peer advertises only implemented report styles & the xApp jumps to a null handler; reproduced independently on radio hardware \\
D5: an ENUMERATED keeps its width across versions & \textit{ack} decoded as \textit{no-ack}; control applies, requester never told \\
D6: a RAN-function identifier implies its version & a v2 service model bound to a v3 decoder, unchecked \\
D7: a registered E2 node never disappears mid-session & fatal assertion in the RIC on node loss; observed on radio hardware \\
\bottomrule
\end{tabular}
\end{table}

\subsection{Energy-Saving Policy Training in the DT}
\label{sub:rl}
\begin{figure}[t]
\centering
\includegraphics[width=0.8\columnwidth]{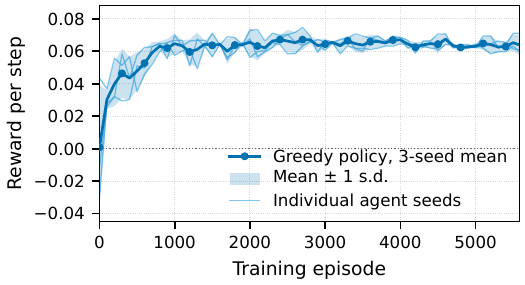}
\caption{Greedy-policy reward on twenty-three held-out DT episodes, as a three-seed mean with $\pm1$ standard deviation.
Training replays ninety-three episodes across sixty passes, so the horizontal axis counts replays.}
\label{fig:rl}
\vspace{-.08in}
\end{figure}
The DT also serves as a training environment, since a cell-sleep policy trains entirely on pre-generated DT episodes and follows the offline protocol of reinforcement-learning work on this simulator~\cite{lacava2023programmable}.
Its reward includes the safe region, so the deployment constraint reaches the policy during training rather than only at admission.
The validation reward rises and then flattens across seeds in Fig.~\ref{fig:rl}, and the learned policy matches a hand-tuned idle-persistence rule at a lower violation rate.
The conformal fidelity gate then screens the learned policy like any other source of actions, admitting one-fifth of its proposals at $\alpha_{\mathrm{gate}}=0.30$ but only 0.3\% of the proposals of a deliberately unsafe policy that sleeps the most loaded cell.
The realized false-approval rate stays inside every budget of that sweep, so the gate remains an independent safety layer.

%% file: sections/conclusion.tex
\section{Conclusion}
\label{sec:conclusion}
In this paper, we introduced \algname, a closed-loop O-RAN DT that tracks the live measurement stream of one deployment and attaches a prescribed error budget to every xApp trust decision at the near-RT RIC.
A drift detector and an adaptive calibrator bound how long the DT can stay out of date, while a conformal fidelity gate screens each action against the operator's safety constraints inside the near-RT decision budget.
Experiments in simulation, on two independently implemented 5G stacks, and on radio hardware confirm single-digit DT error and negligible false-approval rates.

%% file: ref.bib
@book{tripathi2025fundamentals,
  title={Fundamentals of {O-RAN}},
  author={Tripathi, Nishith D. and Shah, Vijay K.},
  year={2025},
  publisher={Wiley},
}

@inproceedings{10.1145/3592149.3592161,
author = {Lacava, Andrea and Bordin, Matteo and Polese, Michele and Sivaraj, Rajarajan and Zugno, Tommaso and Cuomo, Francesca and Melodia, Tommaso},
title = {{ns-O-RAN}: Simulating {O-RAN 5G} Systems in ns-3},
year = {2023},
booktitle = {Proc. Workshop ns-3},
}

@misc{ns-O-RAN-flexric,
  author       = {{Orange Innovation}},
  title        = {{ns-O-RAN-flexric}: {RIC} Testing as a Platform ({RIC-TaaP})},
  year         = {2024},
  howpublished = {\url{https://github.com/Orange-OpenSource/ns-O-RAN-flexric}},
  note         = {Accessed: Aug.\ 2026}
}

@ARTICLE{polese2024colosseum,
  author={Polese, Michele and Bonati, Leonardo and D'Oro, Salvatore and Johari, Pedram and Villa, Davide and Velumani, Sakthivel and Gangula, Rajeev and Tsampazi, Maria and Paul Robinson, Clifton and Gemmi, Gabriele and Lacava, Andrea and Maxenti, Stefano and Cheng, Hai and Melodia, Tommaso},
  journal={IEEE Open J. Commun. Soc.}, 
  title={Colosseum: The Open {RAN} Digital Twin}, 
  year={2024},
}

@inproceedings{schmidt2021flexric,
  author = {Schmidt, Robert and Irazabal, Mikel and Nikaein, Navid},
  title = {{FlexRIC}: An {SDK} for Next-Generation {SD-RANs}},
  booktitle = {Proc. ACM CoNEXT},
  year = {2021},
}

@article{agarwal2025open,
  title={Open {RAN} for {6G} Networks: Architecture, Use Cases and Open Issues},
  author={Agarwal, Bharat and Irmer, Ralf and Lister, David and Muntean, Gabriel-Miro},
  journal={IEEE Commun. Surveys Tuts.},
  year={2026},
}

@techreport{ITUT2022Y3090,
  author       = {{ITU-T-3090}},
  title        = {Digital Twin Network -- Requirements and Architecture},
  institution  = {International Telecommunication Union},
  year         = 2022,
  type         = {Recommendation},
  number       = {Y.3090},
}

@misc{ITU_IMT2030,
  author       = {{ITU-R WP 5D}},
  title        = {{IMT towards 2030 and beyond (IMT-2030)}},
  year         = {2023},
}

@misc{NVIDIA_AerialOmniverse2024,
  author       = {Yang, Jin and Andersin, Michael and Balercia, Tommaso and Chong, CC},
  title        = {Developing Next-Generation Wireless Networks with {NVIDIA} Aerial Omniverse Digital Twin},
  howpublished = {NVIDIA Developer Blog},
  year         = {2024},
}

@article{lacava2023programmable,
  title={Programmable and customized intelligence for traffic steering in {5G} networks using open {RAN} architectures},
  author={Lacava, Andrea and Polese, Michele and Sivaraj, Rajarajan and Soundrarajan, Rahul and Bhati, Bhawani Shanker and Singh, Tarunjeet and Zugno, Tommaso and Cuomo, Francesca and Melodia, Tommaso},
  journal = {IEEE Trans. Mobile Comput.},
  year={2024},
}

@article{zhang2025digital,
  title={Digital Network Twins for Next-Generation Wireless: Creation, Optimization, and Challenges},
  author={Zhang, Zifan and Peng, Zhiyuan and Yu, Hanzhi and Chen, Mingzhe and Liu, Yuchen},
  journal={IEEE Netw.},
  year={2025},
  note={Early access},
}

@article{zhou2025digital,
  title={Digital twins meet open {RAN}: Case studies, implementation, and opportunities},
  author={Zhou, Longyu and Ngo, Mao V and Chen, Binbin and Quek, Tony QS},
  journal={IEEE Commun. Mag.},
  year={2025},
}

@article{owdt2025,
  title={Open Wireless Digital Twin: End-to-End {5G} Mobility Emulation with {OpenAirInterface} and Ray Tracing},
  author={Iye, Tetsuya and Sakamoto, Masaya and Takaya, Shohei and Sato, Eisaku and Susukida, Yuki and Nagaoka, Yu and Maruta, Kazuki and Nakazato, Jin},
  journal={IEEE Access},
  year={2025},
}

@article{tao2025provable,
  title={Provable Performance Bounds for Digital Twin-Driven Reinforcement Learning in Wireless Networks: A Novel Digital-Twin Bisimulation Metric},
  author={Tao, Zhenyu and Xu, Wei and You, Xiaohu},
  journal={IEEE Trans. Signal Process.},
  year={2025},
}

@article{ramdas2023gametheoretic,
  title={Game-Theoretic Statistics and Safe Anytime-Valid Inference},
  author={Ramdas, Aaditya and Gr{\"u}nwald, Peter and Vovk, Vladimir and Shafer, Glenn},
  journal={Statist. Sci.},
  year={2023},
}

@inproceedings{gibbs2021adaptive,
  title={Adaptive Conformal Inference Under Distribution Shift},
  author={Gibbs, Isaac and Cand\`es, Emmanuel},
  booktitle={Proc. NeurIPS},
  year={2021}
}

@article{gibbs2022conformal,
  title={Conformal Inference for Online Prediction with Arbitrary Distribution Shifts},
  author={Gibbs, Isaac and Cand\`es, Emmanuel},
  journal={J. Mach. Learn. Res.},
  year={2024}
}

@inproceedings{angelopoulos2022conformal,
  title={Conformal Risk Control},
  author={Angelopoulos, Anastasios N. and Bates, Stephen and Fisch, Adam and Lei, Lihua and Schuster, Tal},
  booktitle={Proc. ICLR},
  year={2024}
}

@article{shin2022edetectors,
  title={E-detectors: A Nonparametric Framework for Sequential Change Detection},
  author={Shin, Jaehyeok and Ramdas, Aaditya and Rinaldo, Alessandro},
  journal={New Engl. J. Stat. Data Sci.},
  year={2023},
}

@article{simeone2025conformal,
  title={Conformal Calibration: Ensuring the Reliability of Black-Box {AI} in Wireless Systems},
  author={Simeone, Osvaldo and Park, Sangwoo and Zecchin, Matteo},
  journal={IEEE Commun. Mag.},
  year={2026},
}

@article{hou2024automatic,
  title={Automatic {AI} Model Selection for Wireless Systems: Online Learning via Digital Twinning},
  author={Hou, Qiushuo and Zecchin, Matteo and Park, Sangwoo and Cai, Yunlong and Yu, Guanding and Chowdhury, Kaushik and Simeone, Osvaldo},
  journal={IEEE Trans. Wireless Commun.},
  year={2025},
}

@article{jawad2026safetycopilot,
  title   = {A Runtime Safety Copilot for {AI}-Native {O-RAN}: Predictive Verification and Fail-Safe Enforcement in Near-{RT} {RIC} Control Loops},
  author  = {Jawad, Md Asef and Munna, Mohammad Mehedi Hasan and Kabir, Acramul Haque and Antu, Nahid Hossain and Tulona, Reefka Fabliha},
  journal = {IEEE Access},
  year    = {2026},
}

@article{hansen2016cmaes,
  title={The {CMA} Evolution Strategy: A Tutorial},
  author={Hansen, Nikolaus},
  journal={arXiv preprint arXiv:1604.00772},
  year={2016},
}

@article{howard2021timeuniform,
  title={Time-uniform, nonparametric, nonasymptotic confidence sequences},
  author={Howard, Steven R. and Ramdas, Aaditya and McAuliffe, Jon and Sekhon, Jasjeet},
  journal={Ann. Statist.},
  year={2021},
}

@article{taquet2022mapie,
  title={{MAPIE}: an open-source library for distribution-free uncertainty quantification},
  author={Taquet, Vianney and Blot, Vincent and Morzadec, Thomas and Lacombe, Louis and Brunel, Nicolas},
  year={2022},
  journal={arXiv preprint arXiv:2207.12274},
}

@article{ansari2024chronos,
  title={Chronos: Learning the Language of Time Series},
  author={Ansari, Abdul Fatir and Stella, Lorenzo and Turkmen, Caner and Zhang, Xiyuan and Mercado, Pedro and Shen, Huibin and Shchur, Oleksandr and Rangapuram, Syama Sundar and Arango, Sebastian Pineda and Kapoor, Shubham and Zschiegner, Jasper and Maddix, Danielle C. and Wang, Hao and Mahoney, Michael W. and Torkkola, Kari and Wilson, Andrew Gordon and Bohlke-Schneider, Michael and Wang, Yuyang},
  journal={Trans. Mach. Learn. Res.},
  year={2024},
  note={arXiv:2403.07815}
}

@inproceedings{deshpande2025twinran,
  title={{TwinRAN}: Twinning the {5G} {RAN} in {Azure} Cloud},
  author={Deshpande, Yash and Sulkaj, Eni and Kellerer, Wolfgang},
  booktitle={Proc. IEEE/IFIP NOMS},
  year={2025}
}

@inproceedings{zhang2026untwinning,
  title={Network Digital Untwinning: Towards Backward Optimization of Digital Twins},
  author={Zhang, Zifan and Chen, Dianwei and Gao, Anjun and Wang, Manhua and Chen, Mingzhe and Fang, Minghong and Yang, Xianfeng and Liu, Yuchen},
  booktitle={Proc. IEEE ICDCS},
  year={2026},
}

@article{kaltenberger2020oai,
  title={{OpenAirInterface}: Democratizing innovation in the {5G} era},
  author={Kaltenberger, Florian and Silva, Aloizio P. and Gosain, Abhimanyu and Wang, Luhan and Nguyen, Tien-Thinh},
  journal={Comput. Netw.},
  year={2020},
}

@article{breen2021powder,
  title={{POWDER}: Platform for open wireless data-driven experimental research},
  author={Breen, Joe and Buffmire, Andrew and Duerig, Jonathon and Dutt, Kevin and Eide, Eric and Ghosh, Anneswa and Hibler, Mike and Johnson, David and Kasera, Sneha Kumar and Lewis, Earl and Maas, Dustin and Martin, Caleb and Orange, Alex and Patwari, Neal and Reading, Daniel and Ricci, Robert and Schurig, David and Stoller, Leigh B. and Todd, Allison and Van der Merwe, Jacobus and Viswanathan, Naren and Webb, Kirk and Wong, Gary},
  journal={Comput. Netw.},
  year={2021},
}

@misc{srsran2025project,
  author       = {{Software Radio Systems}},
  title        = {{srsRAN Project}: Open-Source {O-RAN} {5G} {CU/DU}},
  year         = {2025},
  howpublished = {\url{https://github.com/srsran/srsRAN_Project}}
}

@article{lopezperez2022energy,
  title={A survey on {5G} radio access network energy efficiency: Massive {MIMO}, lean carrier design, sleep modes, and machine learning},
  author={L{\'o}pez-P{\'e}rez, David and De Domenico, Antonio and Piovesan, Nicola and Geng, Xinli and Bao, Harvey and Song, Qitao and Debbah, M{\'e}rouane},
  journal={IEEE Commun. Surveys Tuts.},
  year={2022},
}

@article{xie2020projected,
  title={Bayesian projected calibration of computer models},
  author={Xie, Fangzheng and Xu, Yanxun},
  journal={J. Amer. Statist. Assoc.},
  year={2020},
}

@misc{orabona2026online,
  author = {Orabona, Francesco},
  title  = {Online Learning: A Modern Introduction Using Convex Optimization},
  year   = {2026},
  note   = {arXiv:1912.13213v10}
}

@inproceedings{chen2016xgboost,
  title={{XGBoost}: A Scalable Tree Boosting System},
  author={Chen, Tianqi and Guestrin, Carlos},
  booktitle={Proc. ACM KDD},
  year={2016}
}
